 \documentclass[preprint,5p,twocolumn,authoryear]{elsarticle}

\usepackage[T1]{fontenc}
\usepackage[latin9]{inputenc}
\usepackage{float}
\usepackage{bm}
\usepackage{amsmath}
\usepackage{amsthm}
\usepackage{amssymb}
\usepackage{graphicx}
\usepackage{color}
\usepackage{lineno,hyperref}
\usepackage{natbib}

\usepackage[latin9]{inputenc}
\usepackage{amsmath}
\usepackage{amsthm}
\usepackage{amssymb}
\usepackage{stmaryrd}
\usepackage{graphicx}

\usepackage[T1]{fontenc}
\usepackage[latin9]{inputenc}
\usepackage{float}
\usepackage{bm}
\usepackage{amsmath}
\usepackage{amsthm}
\usepackage{amssymb}
\usepackage{graphicx}
\usepackage{color}
\usepackage{lineno,hyperref}
\usepackage{natbib}
\usepackage{wrapfig}
\makeatletter
\providecommand{\tabularnewline}{\\}
\floatstyle{ruled}
\newfloat{algorithm}{tbp}{loa}
\providecommand{\algorithmname}{Algorithm}
\floatname{algorithm}{\protect\algorithmname}

\newtheorem{thm}{Theorem}
\newdefinition{defn}{Definition}
\newdefinition{problem}{Problem}
\newtheorem{lem}{Lemma}
\newdefinition{prop}{Proposition}
\newdefinition{example}{Example}
\newdefinition{rem}{Remark}
\newdefinition{fact}{Assumption}

\makeatother

\journal{...}
\begin{document}
	
	\begin{frontmatter}
\title{Periodic event-triggered output regulation for linear multi-agent systems\\  \textcolor{blue}{Extended version (Supplementary file)}}

\author[hz,aa]{Shiqi Zheng\corref{cor1}}



\address[hz]{School of Automation, China University of Geosciences,Wuhan 430074, China}
\address[aa]{Hubei key Laboratory of Advanced Control and Intelligent Automation for Complex Systems, Wuhan 430074, China}

\author[l2]{Peng Shi}
\address[l2]{School of Electrical and Electronic Engineering, The University of Adelaide, Adelaide, SA 5005, Australia}

\author[l3]{Ramesh K. Agarwal}
\address[l3]{Department of Mechanical Engineering, Washington University in St Louis Campus, Box 1185, MO 63130, USA}

\author[l4]{Chee Peng Lim}
\address[l4]{Institute for Intelligent Systems Research and Innovation, Deakin University, Australia}

\begin{abstract}
This study considers the problem of periodic event-triggered (PET)
cooperative output regulation for a class of linear multi-agent systems.
The advantage of the PET output regulation is that the data transmission
and triggered condition are only needed to be monitored at discrete
sampling instants. It is assumed that only a small number of agents
can have access to the system matrix and states of the leader. Meanwhile,
the PET mechanism is considered not only in the communication between
various agents, but also in the sensor-to-controller and controller-to-actuator
transmission channels for each agent. The above problem set-up will
bring some challenges to the controller design and stability analysis.
Based on a novel PET distributed observer, a PET dynamic output feedback
control method is developed for each follower. Compared with the existing
works, our method can naturally exclude the Zeno behavior, and the
inter-event time becomes multiples of the sampling period. Furthermore,
for every follower, the minimum inter-event time can be determined
\textit{a prior}, and computed directly without the knowledge of the
leader information. An example is given to verify and illustrate the
effectiveness of the new design scheme. 
\end{abstract}


	\end{frontmatter}


\section{Introduction}

Cooperative control for multi-agent systems has attracted extensive
attention because of its potential applications \citep{key-3,key-b2,key-6aa,key-icl, key-b33, key-17}
in multi-vehicle formation, wireless sensor network, electrical power
systems $etc$. The cooperative control problem includes leaderless
and leader-following consensus, containment, rendezvous, formation
$etc$. Various control strategies have been utilized for multi-agent
systems, such as adaptive control \citep{key-li-2,key-b2,key-15,key-inn},
sliding mode control \citep{key-2a} and model predictive control.

The output regulation problem for multi-agent systems has recently
drawn much interest from researchers. The purpose of the regulation
problem is to make the output of each follower track a class of reference
input and simultaneously handle the external disturbance \citep{key-14a}.
The reference input and disturbance signals are generated by the exosystem
or leader. In this sense, the output regulation problem is more general
than the standard tracking or stabilization problem \citep{key-9-2,key-li-1,key-7a,key-16a}.
Until now, many excellent results have been proposed in this field
\citep{key-13,key-ad1,key-14a,key-6aaa,key-3}. For instance, in \cite{key-14aa},
based on a new adaptive distributed observer, the cooperative output
regulation problem for linear multi-agent systems was solved.  \cite{key-2}
studied the semi-global output feedback regulation problem for a class
of nonlinear multi-agent systems with heterogeneous relative degrees.

Most of the above works attempt to solve the output regulation problem
under the assumption that all the states can be transmitted continuously.
However, continuous transmission can entail high communication cost
and energy consumption. As a solution to this issue, event-triggered
control strategies have been presented \citep{key-14aa-1,key-5,key-6,key-3}.
The idea of event-triggered control is to transmit the data according
to a well-defined triggered condition. In this way, the communication
burden can be reduced considerably. 

Different types of event-triggered mechanisms \citep{key-2a-1} have
been proposed such as continuous-time event-triggered control, self-triggered
control, dynamic event-triggered control $etc$. More recently, the periodic
event-triggered (PET) control strategy has become a hot topic \citep{key-8,key-11,key-7,key-4}.
Compared with other event-triggered mechanisms, the key feature of
PET control is that the data transmission and the triggered condition
are only needed to be monitored at discrete sampling instants. This
feature benefits control systems in the following aspects: 1) It naturally
rules out the Zeno behavior; 2) The inter-event times become multiples
of the sampling period. This can be very helpful for digital implementation,
and scheduling of  many applications on a shared communication medium;
3) It is more practical in some engineering situations where the states
measurements are only available at periodic intervals due to the constraints
on sensors and network. 4) It can reduce the energy consumption for
evaluating triggered conditions in contrast with continuous-time event-triggered
control. \textit{Nevertheless,} \textit{to the best of our knowledge,
	no works have ever considered the periodic event-triggered output regulation
	for multi-agent systems. }

Motivated by the aforementioned idea, this paper will consider the
PET cooperative output regulation problem for a class of linear multi-agent
systems. The problem is challenging due to the following reasons:
1) Only some of the followers have access to the system matrix or
the states information of the leader; 2) The PET mechanism is considered
not only in the communication between various agents, but also in
the sensor-to-controller and controller-to-actuator transmission channels
for each agent; 3) Only the output information of the followers is
available for the controller design. 

The above problem set up makes the existing output regulation methods
\citep{key-14a,key-9} infeasible. Moreover, directly extending the
distributed observer method \citep{key-16,key-12,key-14aa,key-10}
to our case is not easy because of the PET mechanism. In fact, PET
control is more general than sampled data control. However, the sampled data output regulation problem has not been fully
 investigated so far, not
to mention PET output regulation. This research gap makes stability
analysis challenging.

Our work provides the following main contributions:
\begin{itemize}
	\item Novel PET distributed observers are formulated to estimate the system
	matrix and state information of the leader;
	\item Using the estimated leader information, a new PET dynamic output feedback
	controller is designed for each follower;
	\item Based on the skillful use of some matrix norm and Gronwall's inequalities,
	we prove that the cooperative output regulation problem is solvable
	by the proposed method. 
	\item \textcolor{black}{For each follower, the minimum inter-event time can
		be determined }\textit{\textcolor{black}{a prior}}\textcolor{black}{{}
		and computed directly without the knowledge of the leader information.
		Moreover, by decreasing the gains of the distributed observer, the
		minimum inter-event time for the communication between various agents
		can be made arbitrary long.}
\end{itemize}

The organization of the paper is as follows. Problem formulation and
preliminaries are given in Section 2. The proposed PET distributed
observer and output feedback controller are presented in Section 3.
Simulations are conducted and presented in Section 4. Section 5 concludes the paper. \underline{}

\textit{Notations.} Given a matrix $X_{i}\in\mathbb{R}^{n_{i}\times m}(i=1,2,...,N)$,
$\mathrm{col}(X_{1},X_{2},...,X_{N})=[X_{1}^{\mathrm{T}}\thinspace X_{2}^{\mathrm{T}}\thinspace...X_{N}^{\mathrm{T}}]^{\mathrm{T}}$.
For $A\in\mathbb{R}^{n\times m}$, $\mathrm{vec}(A)=\mathrm{col}(A_{1},A_{2},...,A_{m})$
where $A_{i}\in\mathbb{R}^{n}$ denotes the $i$th column of $A$.
$||A||,||A||_{F}$ are the $2$-norm and Frobenius-norm of matrix
$A$.

\section{Problem formulation and preliminaries}

\subsection{Problem formulation}

Consider a multi-agent system consisting of $N$ followers and $1$
leader. The dynamic of the leader is given by:
\begin{align}
\dot{v} & =Sv\label{eq:1}
\end{align}
where $v\in\mathbb{R}^{n_{v}}$ is the reference input and/or external
disturbance with a positive integer $n_{v}$. $S$ is a given system
matrix. 

The followers are given by the following linear system:
\begin{align}
\dot{x}_{i} & =A_{i}x_{i}+B_{i}u_{i}+E_{i}v,\label{eq:2-1}\\
e_{i} & =C_{i}x_{i}+D_{i}u_{i}+F_{i}v,\label{eq:2}\\
y_{mi} & =C_{mi}x_{i}+D_{mi}u_{i}+F_{mi}v\label{eq:2-2}
\end{align}
where $i\in\{1,2,...,N\}$. $x_{i}\in\mathbb{R}^{n_{i}}$, $u_{i}\in\mathbb{R}^{n_{ui}}$,
$e_{i}\in\mathbb{R}^{n_{ei}}$, $y_{mi}\in\mathbb{R}^{n_{yi}}$ are
the system states, control effort, consensus error and measurement
output respectively with positive integers $n_{i},n_{ui},n_{ei},n_{yi}$.
$A_{i},B_{i},C_{i},D_{i},E_{i},F_{i},C_{mi},D_{mi},F_{mi}$ are the
given system matrices.

The communication for the multi-agent systems is represented by a
directed graph $\mathcal{G}$. Let $\mathcal{G}=(\mathcal{V},\mathcal{E})$
where $\mathcal{V}=\{1,2,...,N\}$ denotes the set of vertices, $\mathcal{E}\subseteq\mathcal{V}\times\mathcal{V}$
the set of edges. Let $\mathcal{N}_{i}$ represents the neighbors
of agent $i$, $i.e.$, $\mathcal{N}_{i}=\{j|j\in\mathcal{V}|(j,i)\in\mathcal{E}\}$.
Define matrix $\mathcal{W}=[a_{ij}]\in\mathbb{R}^{N\times N}$ such
that if $(j,i)\in\mathcal{E}$ then $a_{ij}=1$, otherwise $a_{ij}=0$.
Self-loop is not allowed, $i.e.$, $a_{ii}=0$ for $i\in\mathcal{V}$.
Define Laplacian matrix as $\mathcal{L}=\mathcal{D}-\mathcal{\mathcal{W}}$
with $\mathcal{D}=\mathrm{diag}(d_{1},d_{2},...,d_{N})$ and $d_{i}=\sum_{j\in\mathcal{N}_{i}}a_{ij}(i\in\mathcal{V})$.
For the information transmission between the leader and followers,
define $a_{i0}$ such that if the followers are connected to the leader,
then $a_{i0}=1$; otherwise $a_{i0}=0$. Also let $a_{0i}=0$. Note
that only a small portion of followers have access to the leader.

Based on the above analysis, \textcolor{black}{the cooperative output
	regulation problem} is formulated as follows:
\begin{problem}
	\label{prob:Given-a-multi-agent}Given a multi-agent system (\ref{eq:1})-(\ref{eq:2})
	with its corresponding graph $\mathcal{G}$, develop a PET distributed
	control law for each follower such that 
	
	1) All the closed loop signals are bounded for all $t\in[0,+\infty)$;
	and,
	
2) The output regulation error satisfies $\underset{t\rightarrow+\infty}{\lim}||e_{i}(t)||=0$
or \textcolor{blue}{$\underset{t\rightarrow+\infty}{\lim}||e_{i}(t)||\leq\varLambda_{i}$
	for $i\in\{1,2,...,N\}$ where $\varLambda_{i}$ is a small
	positive constant.}
\end{problem}
\begin{rem}
	As we will see in Section 3, according to whether  the PET mechanism
	is adopted for the controller-to-actuator channel, the\textcolor{black}{{}
		cooperative output regulation error} will be regulated to exact zero
	or a small neighborhood around the origin. In addition, when the limitation of $||e_{i}(t)||$ does not exist, $\underset{t\rightarrow+\infty}{\lim}||e_{i}(t)||$ should be understood as $\underset{t\rightarrow+\infty}{\lim\sup}||e_{i}(t)||$.
\end{rem}
\begin{rem}
	\noindent \textcolor{black}{The regulation error in (\ref{eq:2}) can
		be seen as a generalization of the consensus error defined in many
		literatures \citep{key-9}. For instance, suppose the output of the
		followers is $y_{i}=C_{i}x_{i}$. Then if one wants the followers
		to track the leader, the consensus error may be defined as $y_{i}-v=C_{i}x_{i}-v$.
		This is equivalent to let $D_{i}=0,F_{i}=-I$ in (3). In addition,
		note that the measurement output $y_{mi}$ in (\ref{eq:2-2}) may
		not equal to the real output of the followers. For example, if the
		real output of the followers is $y_{i}=C_{mi}x_{i}+D_{mi}u_{i}$.,
		then the measurement output $y_{mi}$ in (\ref{eq:2-2}) indicates
		that the real output $y_{i}$ may be influenced by an external disturbance
		$F_{mi}v$ where $v$ is generated by the exosystem $v=Sv$.}
\end{rem}

\subsection{Preliminaries}

In this subsection, we will introduce some basic assumptions and results
for the cooperative output regulation problem. It is divided into
three parts.

\textit{1) Graph and leader}

For the communication graph, we assume that:
\begin{fact}
	\label{fact:The-graph-containing}The graph containing the leader
	and $N$ followers has a directed spanning tree with the leader as
	the root. 
\end{fact}

Then we have the following result.
\begin{lem}
	\citep{key-14aa}\label{lem:Under-Assumption-,} Under Assumption \ref{fact:The-graph-containing},
	$-\mathcal{H}$ is a Hurwitz matrix with $\mathcal{H}\triangleq\mathcal{L}+\mathcal{B}$
	and $\mathcal{B}\triangleq\mathrm{diag}(a_{10},a_{20},...,a_{N0})$.
\end{lem}
For the leader (\ref{eq:1}), we assume
\begin{fact}
	\label{fact:4}The leader system is neutrally stable, 
	\emph{i.e.}, the eigenvalues
	of $S$ are semi-simple with zero real parts.
\end{fact}
\begin{rem}
	Under the above assumption, we know that as long as the initial value
	$v(0)$ is bounded, $v(t)$ is bounded on $[0,+\infty)$. Meanwhile,
	a wide class of signals, such as sine and step signals, can be generated
	by the leader system (\ref{eq:1}). In addition, from \cite{key-12},
	without loss of generality $S$ can be selected to be a skew-symmetric
	matrix such that $S^{T}=-S$. 
\end{rem}
\textit{2) Followers}

For linear system (\ref{eq:2}), we make the following assumptions.
\begin{fact}
	For $i=1,2,...,N$, the system matrices satisfy:
	
	1) $(A_{i},B_{i})$ are stabilizable;
	
	2) $(C_{mi},A_{i})$ are detectable;
	
	3) The following linear matrix equations admit a solution $(X_{i},U_{i})$
	\begin{align}
	X_{i}S & =A_{i}X_{i}+B_{i}U_{i}+E_{i},\nonumber \\
	0 & =C_{i}X_{i}+D_{i}U_{i}+F_{i}.\label{eq:4-1}
	\end{align}
\end{fact}

The above assumptions are standard in output regulation theory. Meanwhile,
from \cite{key-14aa}, we know the solution $(X_{i},U_{i})$ can
be solved adaptively. We briefly explain the idea as follows. Let
$\chi_{i}=\mathrm{vec}(\mathrm{col}(X_{i},U_{i}))$, $\beta_{i}=\mathrm{vec}(\mathrm{col}(E_{i},F_{i}))$,
\[
\mathcal{A}_{i}=S^{T}\varotimes\left[\begin{array}{cc}
I & 0\\
0 & 0
\end{array}\right]-I\varotimes\left[\begin{array}{cc}
A_{i} & B_{i}\\
C_{i} & D_{i}
\end{array}\right],
\]
\[
\hat{\mathcal{A}}_{i}(t)=\hat{S}_{i}^{T}(t)\varotimes\left[\begin{array}{cc}
I & 0\\
0 & 0
\end{array}\right]-I\varotimes\left[\begin{array}{cc}
A_{i} & B_{i}\\
C_{i} & D_{i}
\end{array}\right]
\]
where $\hat{\mathcal{A}}_{i}(t),\hat{S}_{i}(t)$ are time-varying
matrices.

Then (\ref{eq:4-1}) can be written as:
\[
\mathcal{A}_{i}\chi_{i}=\beta_{i}.
\]

Define $\hat{\chi}_{i}$ with the adaptive law
\begin{equation}
\dot{\hat{\chi}}_{i}=-\kappa\hat{\mathcal{A}}_{i}^{T}(\hat{\mathcal{A}_{i}}\hat{\chi}_{i}-\beta_{i})\label{eq:5}
\end{equation}
where $\kappa>0$ is a positive design parameter. 

Meanwhile, define the adaptive solution $\hat{X}_{i},\hat{U}_{i}$
such that they have the same dimensions as $X_{i},U_{i}$ and 
\begin{equation}
\mathrm{vec}(\mathrm{col}(\hat{X}_{i},\hat{U}_{i}))=\hat{\chi}_{i}.\label{eq:6}
\end{equation}

Then we have the following result \citep{key-14aa}.
\begin{lem}
	\label{lem:If--converges}If $S-\hat{S}_{i}$ converges to zero exponentially,
	$\chi_{i}-\hat{\chi}_{i}$ and $X_{i}-\hat{X}_{i},U_{i}-\hat{U}_{i}$
	will all converge to zero exponentially.
\end{lem}
\textit{3) Useful inequalities}

Finally, we introduce some inequalities which will be used in the
stability analysis.
\begin{lem}
	\label{lem:(Matrix-norm-inequalities)}(Matrix norm inequalities)
	Given matrices $A,B\in\mathbb{R}^{n\times n}$, we have
	
	1) $||\mathrm{e}^{A+B}-\mathrm{e}^{A}||\leq\mathrm{e}^{||A||+||B||}||B||$;
	
	2) $||A||\leq||A||_{F}\leq\sqrt{n}||A||$; 
	
	3) $||\mathrm{e}^{A}||\leq\mathrm{e}^{||A||}$.
\end{lem}
\begin{proof}
	1) is from Lemma 1 in \cite{key-14}. 2) and 3) can be proved using
	basic matrix theory.
\end{proof}
\begin{lem}
	\label{lem:1)-(Gronwall's-inequality)}(Gronwall's inequality) Given
	a real-valued function $w(t):[0,+\infty)\rightarrow\mathbb{R}$ such
	that
	\[
	w(t)\leq\alpha+\int_{t_{0}}^{t}\beta w(\tau)d\tau
	\]
	for $\forall t\in[t_{0},+\infty)$ where $\alpha,\beta,t_{0}>0$ are
	positive constants. Then 
	\[
	w(t)\leq\alpha\mathrm{e}^{\beta(t-t_{0})}.
	\]
\end{lem}

\section{Main results}

In this section, we will discuss the output regulation problem for
linear multi-agent systems by (\ref{eq:1})-(\ref{eq:2}). The control
scheme is shown in Fig. \ref{fig:3-1}. For the communication between
various agents, the controller of agent $i$ will send/receive the
information to/from its neighbors based on the PET Mechanism A (PETM-A).
For the sensor-to-controller channel, the sensor will sample the output
information from the plant and transmit it to the controller by the
PET Mechanism B (PETM-B). For the controller-to-actuator channel,
two different situations will be considered. We will first consider
the situation where the transmission is continuous, $i.e.$, the switch
in Fig. \ref{fig:3-1} is on node 1. Then, we will consider the case
when the switch is on node 2, that is the control signal will be transmitted
to the actuator based on the PET Mechanism C (PETM-C). We can see
that the PET mechanisms are not only used for the communication between
various agents, but also for the sensor-to-controller and controller-to-actuator
channel in each agent. 

The proposed controller is composed of two parts: a PET distributed
observer and a PET control law. The PET distributed observer is used
to estimate the system matrix $S$ and $v$ of the leader based on
PETM-A. The control law will use the estimated information to generate
the control signal according to PETM-B and PETM-C.

Next, we will explain these two parts respectively.
\begin{figure}
	\begin{centering}
		\includegraphics[scale=0.9]{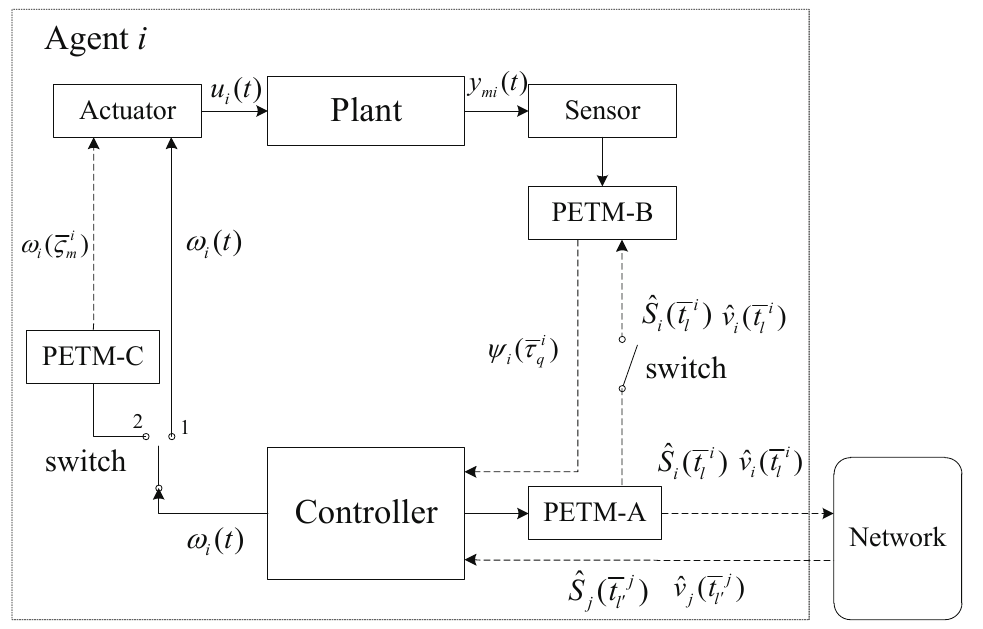}
		\par\end{centering}
	\caption{\label{fig:3-1}Control scheme.}
\end{figure}

\begin{rem}
	\textcolor{black}{Notably the above control scheme indicates that the PET
		mechanism is considered not only in the communication between various
		agents, but also in the sensor-to-controller and controller-to-actuator
		transmission channels for each agent. The motivation for considering
		this control scheme is that in some situations the control of a single
		agent may require the network communication between the controller
		and the sensor/actuator. A number of applications may involve the
		formation control of unmanned automobiles where control of each automobile
		is based on network \citep{key-14-1}, the networked control of a group
		of UAVs based on remote controllers/ground bases \citep{key-10-1,key-10aa},
		the cooperative control of robot manipulators $etc$.}
\end{rem}

\subsection{PET distributed observer}

Let $0=t_{0}<t_{1}<\cdots<t_{k}<\cdots$ denote the sampling time
instants for the multi-agent systems where $t_{k}\triangleq kT$ and
$T>0$ is the sampling period. Also define set $\Omega_{T}=\{t_{0},t_{1},...,t_{k},...\}.$
On each time interval $[t_{k},t_{k+1})$, the distributed observer
for agent $i(i=1,2,...,N)$ is designed as:
\begin{align}
\dot{\hat{S}}_{i} & =\mu_{1}\sum_{j=0}^{N}a_{ij}(\hat{S}_{j}(\overline{t}_{l'}^{j})-\hat{S}_{i}(\overline{t}_{l}^{i})),\label{eq:5-1-1}\\
\dot{\hat{v}}_{i} & =\hat{S}_{i}(\overline{t}_{l}^{i})\hat{v}_{i}(t)+\mu_{2}\sum_{j=0}^{N}a_{ij}(\overline{v}_{j}(t,\overline{t}_{l'}^{j})-\overline{v}_{i}(t,\overline{t}_{l}^{i}))\label{eq:6-1}
\end{align}
where $\mu_{1},\mu_{2}>0$ are two positive parameters. $\hat{S}_{0}(t)\equiv S$,
\begin{equation}
\overline{v}_{i}(t,\overline{t}_{l}^{i})=\mathrm{e}^{\hat{S}_{i}(\overline{t}_{l}^{i})(t-\overline{t}_{l}^{i})}\hat{v}_{i}(\overline{t}_{l}^{i})(i=1,...,N)\label{eq:7-2-1}
\end{equation}
\textcolor{black}{and $\hat{v}_{0}(t)\triangleq v$, $\overline{v}_{0}(t,t_{l}^{0})\triangleq\mathrm{e}^{S(t-t_{l}^{0})}\hat{v}_{0}(t_{l}^{0})= v(t)$. }

Note that $0=\overline{t}_{0}^{i}<\overline{t}_{1}^{i}<\cdots<\overline{t}_{l}^{i}<\cdots$
denote the event-triggered time instants. On time instant $\overline{t}_{l}^{i}$,
agent $i$ will send $\hat{S}_{i}(\overline{t}_{l}^{i})$ and $\hat{v}_{i}(\overline{t}_{l}^{i})$
to its neighbors. The event-triggered time instants are determined
by  PETM-A in Fig. \ref{fig:3-1} which is given by:
\begin{equation}
\overline{t}_{l+1}^{i}=\mathrm{inf}\{\tau>\overline{t}_{l}^{i}|\tau\in\Omega_{T},f_{S}^{i}(\cdot)>0,f_{v}^{i}(\cdot)>0\}\label{eq:26}
\end{equation}
where 
\begin{equation}
f_{S}^{i}(\tau,\overline{t}_{l}^{i})=||\hat{S}_{i}(\tau)-\hat{S}_{i}(\overline{t}_{l}^{i})||_{F}-\iota_{S}\mathrm{e}^{-\gamma_{S}\tau},\label{eq:9}
\end{equation}
\begin{equation}
f_{v}^{i}(\tau,\overline{t}_{l}^{i})=||\hat{v}_{i}(\tau)-\overline{v}_{i}(\tau,\overline{t}_{l}^{i})||-\iota_{v}\mathrm{e}^{-\gamma_{v}\tau}\label{eq:10-3}
\end{equation}
with positive constants $\iota_{S},\iota_{v},\gamma_{S},\gamma_{v}>0$.

It can be seen that the set $\Omega_{ET}^{i}\triangleq\{\overline{t}_{0}^{i},\overline{t}_{1}^{i},...,\overline{t}_{l}^{i},...\}\subseteq\Omega_{T}.$
Meanwhile, in the observer (\ref{eq:5-1-1}) and (\ref{eq:6-1}),
with a little abuse of notation, $\overline{t}_{l}^{i}$ and $\overline{t}_{l'}^{j}$
denote the latest event-triggered time instants for agent $i$ and
$j$ on $[t_{k},t_{k+1})$.

Then, we have the following result.
\begin{thm}
	\label{thm:1}Given a multi-agent system with the leader (\ref{eq:1}),
	then there exists a PET distributed observer in the form of (\ref{eq:5-1-1})-(\ref{eq:26})
	such that $\tilde{S}_{i}\triangleq\hat{S}_{i}-S$ and $\tilde{v}_{i}\triangleq\hat{v}_{i}-v(i=1,2,...,N)$
	converge to zero exponentially. 
\end{thm}
\begin{proof}
	The proof is presented in Appendix A.
\end{proof}
\begin{rem}
	\textcolor{black}{\label{rem:According-to-Appendix-1}(\ref{eq:5-1-1})
		can be expressed as
		\begin{align*}
		\dot{\hat{S}}_{i} & =\mu_{1}\sum_{j=1}^{N}a_{ij}(\hat{S}_{j}(\overline{t}_{l'}^{j})-\hat{S}_{i}(\overline{t}_{l}^{i}))+a_{i0}(\hat{S}_{0}(\overline{t}_{l'}^{0})-\hat{S}_{i}(\overline{t}_{l}^{i})).
		\end{align*}
		For those followers that have access to the leader, we have $a_{i0}=1$.
		Then we define $\hat{S}_{0}(t)\triangleq S$. For those followers
		that do not have access to the leader, $a_{i0}=0$. Therefore, the
		term $a_{i0}(\hat{S}_{0}(\overline{t}_{l'}^{j})-\hat{S}_{i}(\overline{t}_{l}^{i}))$
		in the above equation vanishes, which means the information of the
		leader is not used. Therefore, only a small number of followers have
		access to the system matrix $S$ of the leader. A similar idea can
		be found for (\ref{eq:6-1}).}
	
	\textcolor{black}{It should be noted that the followers know their
		own system matrices. Specifically, (\ref{eq:2-1})-(\ref{eq:2-2})
		are regarded as the system model of the followers. Therefore, each
		follower knows its own system matrices  $E_{i},F_{i}$ and $F_{mi}$.}
	  \textcolor{black}{Also note that in some situations, }    {\textcolor{black}{$E_{i},F_{i}$ and $F_{mi}$}}\textcolor{black}{{}
		can be regarded as the system matrices of the leader. In these situations,
		we can revise the distributed observer (\ref{eq:5-1-1}) to estimate
	}{\textcolor{black}{$E_{i},F_{i}$ and $F_{mi}$}}\textcolor{black}{{}
		similarly.}
\end{rem}
\begin{rem}
	\label{rem:According-to-Appendix}According to Appendix A, one possible
	choice of the sampling period $T$ is to satisfy
	\begin{equation}
	0<T<\frac{1}{\mu_{\max}\left(||P\mathcal{H}||+1\right)||\mathcal{H}||}\label{eq:10-1}
	\end{equation}
	where $\mu_{\max}=\max\{\mu_{1},\mu_{2}\}$. $P$ is design matrix
	such that $P\mathcal{H}+\mathcal{H}^{T}P=2I$. Note that $P$ always
	exists due to $-\mathcal{H}$ is Hurwitz. It is also noted that the
	selection of $T$ is only dependent on the graph information not on
	the matrix $S$. \textcolor{black}{Moreover, we can see that when $\mu_{1}$ and $\mu_{2}$
		are small enough, the sampling period $T$ can be arbitrary long.
		This implies that by decreasing the values of $\mu_{1}$ and $\mu_{2}$,
		the communication burden between various agents can be reduced considerably. }
	
	\textcolor{black}{Note that for PET control, the minimum inter-event
		time is equal to the sampling period. Hence, it can be determined
		explicitly by (\ref{eq:10-1}).}
\end{rem}
\begin{rem}
	\textcolor{black}{\label{rem:Note-that-theoretically,}Note that theoretically
		if $\mu_{1}$ and $\mu_{2}$ and $T$ satisfy (\ref{eq:10-1}), Theorem \ref{thm:1}
		will hold. However, different values of $\mu_{1},\mu_{2}$ can result
		in different control performance. Herein, we give some guidelines
		for the selection of $\mu_{1}$ and $\mu_{2}$. Basically, a larger $\mu_{1}$ and $\mu_{2}$
		will result in a quicker estimation of the leader information $S$ and $v$.
		This may lead to a faster convergence rate for the multi-agent systems.
		However, $\mu_{1}$ and $\mu_{2}$ cannot be selected to be too large mainly
		because of two factors. First, as stated in Remark \ref{rem:According-to-Appendix},
		small $\mu_{1}$ and $\mu_{2}$ can reduce the communication burden and energy
		consumption. Second, the convergence rate of the multi-agent systems
		will not increase too much when $\mu_{1}$ and $\mu_{2}$ are sufficiently
		large. Moreover, when $\mu_{1}$ and $\mu_{2}$ are large, the injection
		terms on the right hand side of the distributed observer (\ref{eq:5-1-1}),
		(\ref{eq:6-1}) may become large and oscillate. This implies that
		more energy will be needed to realize the distributed observer. }
\end{rem}

\subsection{PET control law when PETM-C is not invoked}

In this section, we will design an output feedback controller for
each follower $i\in\{1,2,...,N\}$. We assume that the data transmission
in the controller-to-actuator channel is continuous, $i.e.$, PETM-C is
not invoked. 

The sampling instants for the output information are denoted as $0=\tau_{0}^{i}<\tau_{1}^{i}<\cdots<\tau_{p}^{i}<\cdots$
with sampling period $\mathcal{T}^{i}=\tau_{p+1}^{i}-\tau_{p}^{i}$.
Note that $\tau_{p}^{i}$ can be different with the communication
sampling instants $t_{k}$. Then during time interval $[\tau_{p}^{i},\tau_{p+1}^{i})$,
the output feedback controller is given by:
\begin{align}
u_{i}= & \omega_{i}(t)\label{eq:10}\\
\omega_{i}= & K_{i}\hat{x}_{i}+(\hat{U}_{i}-K_{i}\hat{X}_{i})\hat{v}_{i},\label{eq:12-2}\\
\dot{\hat{x}}_{i}= & A_{i}\hat{x}_{i}+B_{i}\omega_{i}+E_{i}\hat{v}_{i}\nonumber \\
& +L_{i}(C_{mi}\hat{x}_{i}(\tau_{p}^{i})+D_{mi}\omega_{i}(\tau_{p}^{i}))\nonumber \\
& +L_{i}(\rho_{i}F_{i}\overline{v}_{i}(\tau_{p}^{i},\overline{t}_{l}^{i})+(1-\sigma_{i})F_{mi}\overline{v}_{i}(\tau_{p}^{i},\overline{t}_{l}^{i}))\nonumber \\
& +L_{i}\psi_{i}(\overline{\tau}_{q}^{i})\label{eq:11}
\end{align}
where
\begin{align}
\psi_{i}(t)= & \sigma_{i}F_{mi}\overline{v}_{i}(t,\overline{t}_{l}^{i})-\rho_{i}F_{i}\overline{v}_{i}(t,\overline{t}_{l}^{i})-y_{mi}(t),\label{eq:12}\\
\overline{v}_{i}(t,\overline{t}_{l}^{i})= & \mathrm{e}^{\hat{S}_{i}(\overline{t}_{l}^{i})(t-\overline{t}_{l}^{i})}\hat{v}_{i}(\overline{t}_{l}^{i}),
\end{align}
$\hat{U}_{i},\hat{X}_{i}$ are the adaptive solution for the regulator
equation (\ref{eq:4-1}). It is determined by (\ref{eq:5})-(\ref{eq:6})
where $\hat{S}_{i}$ is obtained by the distributed observer. $K_{i},L_{i}$
are selected such that $A_{i}+B_{i}K_{i}$ and $A_{i}+L_{i}C_{mi}$
are both Hurwitz. $\sigma_{i},\rho_{i}$ are two non-negative design
parameters.

$0=\overline{\tau}_{0}^{i}<\overline{\tau}_{1}^{i}<\cdots<\overline{\tau}_{q}^{i}<\cdots$
are the PET instants for agent $i$. On time instant $\overline{\tau}_{q}^{i}$,
the sensor will transmit $\psi_{i}(\overline{\tau}_{q}^{i})$ to the
controller. They are determined by  PETM-B in Fig. \ref{fig:3-1}
which is described as:
\begin{equation}
\overline{\tau}_{p+1}^{i}=\mathrm{inf}\{\tau>\overline{\tau}_{q}^{i}|\tau\in\Omega_{T},f_{\psi}^{i}(\cdot)>0\}\label{eq:26-1-1}
\end{equation}
where 
\begin{equation}
f_{\psi}^{i}(\tau,\overline{\tau}_{q}^{i})=||\psi_{i}(\tau)-\psi_{i}(\overline{\tau}_{q}^{i})||-\iota_{\psi}\mathrm{e}^{-\gamma_{\psi}\tau}\label{eq:17-2}
\end{equation}
with positive constants $\iota_{\psi},\gamma_{\psi}>0$. 

\textcolor{black}{Note that the designed controller (\ref{eq:10})-(\ref{eq:11})
	only uses the estimated information $\hat{S}_{i},\hat{v}_{i}$ from
	the distributed observer. This implies that the proposed control scheme
	satisfies the condition that only a small number of followers have
	access to the leader information.}

\textcolor{black}{Based on the above analysis, we have the following
	result.}
\begin{thm}
	\label{thm:1-1}Given a multi-agent system (\ref{eq:1})-(\ref{eq:2}),
	then there exists a PET output feedback control law in the form of
	(\ref{eq:10})-(\ref{eq:26-1-1}) with PET distributed observer (\ref{eq:5-1-1})-(\ref{eq:26})
	such that Problem \ref{prob:Given-a-multi-agent} is solved with $\underset{t\rightarrow+\infty}{\lim}||e_{i}(t)||=0$.
\end{thm}
\begin{proof}
	The proof is presented in Appendix C.
\end{proof}
\begin{rem}
	\label{rem:From-Appendix-C,}From Appendix C, we know one possible
	selection of sampling period $\mathcal{T}^{i}$ is 
	\begin{equation}
	0<\frac{||Q_{i}L_{i}C_{mi}||\delta_{i3}(\mathcal{T}^{i})}{1-\delta_{i3}(\mathcal{T}^{i})}<1\label{eq:16-1}
	\end{equation}
	where $\delta_{i3}(\mathcal{T}^{i})$ is a $\mathcal{K}$-class function
	such that $\delta_{i3}(\mathcal{T}^{i})=\mathcal{T}^{i}||A_{i}+L_{i}C_{mi}||\mathrm{e}^{||A_{i}||\mathcal{T}^{i}}$,
	$Q_{i}$ is a given matrix such that $Q_{i}(A_{i}+L_{i}C_{mi})+(A_{i}+L_{i}C_{mi})^{\mathrm{\mathit{T}}}Q_{i}=-2I$. 
	
	\textcolor{black}{Note that since $\delta_{i3}(\mathcal{T}^{i})$ is
		a $\mathcal{K}$-class function, there must exist a positive $\mathcal{T}^{i}$
		such that (\ref{eq:16-1}) holds. In fact, one can compute the Maximum
		Allowable Sampling Period for $\mathcal{T}^{i}$ from (\ref{eq:16-1}).
		Meanwhile, similar to Remark \ref{rem:According-to-Appendix}, the
		minimum inter-event time for PET control can be determined by solving
		(\ref{eq:16-1}) numerically.}
\end{rem}
\begin{rem}
	\textcolor{black}{\label{rem:For-the-parameters}For the design parameters
		in the event-triggered mechanisms (\ref{eq:9}), (\ref{eq:10-3})
		and (\ref{eq:17-2}), larger $\iota_{S},\iota_{v}\iota_{\psi}$ and
		smaller $\gamma_{S},\gamma_{v}\gamma_{\psi}$ indicate a larger threshold
		for the event-triggered mechanism. Hence, more communication burden
		could be reduced. However, for larger $\iota_{S},\iota_{v}\iota_{\psi}$
		and smaller $\gamma_{S},\gamma_{v}\gamma_{\psi}$, the threshold values
		will take more time to converge to zero. This indicates the convergence
		speed for the multi-agent systems may become slower. }
	
	The selection of the design parameters $\sigma_{i}$ and $\rho_{i}$
	in (\ref{eq:11}) is flexible. Basically they can be any positive
	real numbers. However, different values of $\sigma_{i}$ and $\rho_{i}$
	can result in different control performance and communication burden.
	Specifically, $\sigma_{i}$ and $\rho_{i}$ should be selected such that
	when $t\rightarrow+\infty$, $\psi_{i}(t)$ can be as small as possible.
	For typical example, if $y_{mi}$ can be expressed as $y_{mi}=C_{i}x_{i}+D_{i}u_{i}+F_{mi}v$,
	then $\sigma_{i}=\rho_{i}=1$. In this case, by (\ref{eq:2}), we
	have
	\begin{align*}
	\psi_{i}(t)= & F_{mi}\overline{v}_{i}(t,\overline{t}_{l}^{i})-F_{i}\overline{v}_{i}(t,\overline{t}_{l}^{i})-y_{mi}(t)\\
	= & -(C_{i}x_{i}+D_{i}u_{i}+F_{i}v)\\
	& +\left(F_{mi}-F_{i}\right)\left(\overline{v}_{i}(t,\overline{t}_{l}^{i})-v\right)\\
	= & -e_{i}+\left(F_{mi}-F_{i}\right)\left(\overline{v}_{i}(t,\overline{t}_{l}^{i})-v\right).
	\end{align*}
	According to Theorems \ref{thm:1}-\ref{thm:1-1}, we know $\left(\overline{v}_{i}(t,\overline{t}_{l}^{i})-v\right)\rightarrow0$,
	$e_{i}\rightarrow0$ as $t\rightarrow+\infty$. Therefore, we can
	conclude that $\psi_{i}\rightarrow0$, as $t\rightarrow+\infty$.
	Hence, $\psi_{i}$ is minimized as $t\rightarrow+\infty$. 
\end{rem}
\begin{rem}
	\textcolor{black}{From Fig. \ref{fig:3-1} and (\ref{eq:12}), we know
		PETM-B may need some information about the leader, $i.e.$, $\hat{v}_{i}(\overline{t}_{l}^{i})$ and $\hat{S}_{i}(\overline{t}_{l}^{i})$.
		This information can be transmitted by the controller. Note that the
		information $\hat{v}_{i}(\overline{t}_{l}^{i})$ and $\hat{S}_{i}(\overline{t}_{l}^{i})$
		is not necessary for  PETM-B since one can set $\sigma_{i}=\rho=0$
		(the single switch in Fig. \ref{fig:3-1} is off). However, as stated in
		Remark \ref{rem:For-the-parameters}, using this information, $i.e.,$
		set $\sigma_{i}=\rho=1$, the communication burden can be reduced
		considerably. ,}
	
	\textcolor{black}{One may wonder whether the communication burden
		between the controller and sensor may increase if the controller sends
		some information to  PETM-B. As described in Remark \ref{rem:According-to-Appendix},
		we know the communication burden for transmitting $\hat{v}_{i}(\overline{t}_{l}^{i})$ and $\hat{S}_{i}(\overline{t}_{l}^{i})$
		can be very small because the sampling time $T$ can be made arbitrary
		long by tuning the control parameters $\mu_{1}$ and $\mu_{2}$. Therefore,
		the overall communication burden for the sensor-to-controller channel
		can still be reduced. In addition, there are several alternative ways
		to remove the communication from the controller to the sensor. Please
		see Appendix E for details.}
\end{rem}

\subsection{PET control law when PETM-C is invoked}

Let us consider the case when the PET mechanism is used in the controller-to-actuator
channel, $i.e.$,  PETM-C is invoked in Fig. \ref{fig:3-1}. It
is noted that since we consider a regulation problem, the tracking
error may not converge to exact zero because of the discrete transmission.
This is a common case for a regulation or tracking problem (see \cite{key-6a,key-7a}).
In fact, the error will be regulated to an arbitrary small neighborhood
around the origin similar to \cite{key-6a}. 

Consider the following control law
\begin{align}
u_{i}(t)= & \omega_{i}(\overline{\varsigma}_{m}^{i}),\thinspace t\in[\overline{\varsigma}_{m}^{i},\overline{\varsigma}_{m+1}^{i})\label{eq:10-2}
\end{align}
where $\omega_{i}$ is described by (\ref{eq:12-2})-(\ref{eq:26-1-1})
except the triggered function $f_{\psi}^{i}(\tau,\overline{\tau}_{q}^{i})$
is modified as:
\begin{equation}
\overline{f}_{\psi}^{i}(\tau,\overline{\tau}_{q}^{i})=||\psi_{i}(\tau)-\psi_{i}(\overline{\tau}_{q}^{i})||-\left(\iota_{\psi}\mathrm{e}^{-\gamma_{\psi}\tau}+\overline{\iota}_{\psi}\right)\label{eq:17-2-1}
\end{equation}
with constants $\iota_{\psi},\overline{\iota}_{\psi},\gamma_{\psi}\geq0$. 

$0=\overline{\varsigma}_{0}^{i}<\overline{\varsigma}_{1}^{i}<\cdots<\overline{\varsigma}_{m}^{i}<\cdots$
is the the PET instants for the controller-to-actuator channel. On
time instant $\overline{\varsigma}_{m}^{i}$, the controller will
transmit $\omega_{i}(\overline{\varsigma}_{m}^{i})$ to the actuator.
They are described as:
\begin{equation}
\overline{\varsigma}_{m+1}^{i}=\mathrm{inf}\{\tau>\overline{\varsigma}_{m}^{i}|\tau\in\Omega_{\mathcal{T}^{i}},f_{\omega}^{i}(\cdot)>0\}\label{eq:26-1-1-1}
\end{equation}
where 
\begin{equation}
f_{\omega}^{i}(\tau,\overline{\varsigma}_{m}^{i})=||\omega_{i}(\tau)-\omega_{i}(\overline{\varsigma}_{m}^{i})||-\left(\iota_{\omega}\mathrm{e}^{-\gamma_{\omega}\tau}+\overline{\iota}_{\omega}\right)\label{eq:26-1-1-1-1}
\end{equation}
with constants $\iota_{\omega},\overline{\iota}_{\omega},\gamma_{\omega}\geq0$. 

Then, we have the following result.
\begin{thm}
	\label{thm:2}Given a multi-agent system (\ref{eq:1})-(\ref{eq:2}),
	then there exists a PET output feedback control law in the form of
	(\ref{eq:10-2})-(\ref{eq:26-1-1-1}) with PET distributed observer
	(\ref{eq:5-1-1})-(\ref{eq:26}) such that Problem \ref{prob:Given-a-multi-agent}
	is solved with \textcolor{blue}{$\underset{t\rightarrow+\infty}{\lim}||e_{i}(t)||\leq\varLambda_i$. }
\end{thm}
\begin{proof}
	The proof is presented in Appendix D.
\end{proof}
\begin{rem}
	\textcolor{black}{\label{rem:11}\textcolor{blue}{From Appendix D, we know for agent
			$i\in\{1,2,...,N\}$, there exists a non-negative increasing function
			$\varphi_{i8}(\overline{\iota}_{\psi},\overline{\iota}_{\omega},\mathcal{T}^{i})$
			with $\varphi_{i8}(0,0,0)=0$ such that $\underset{t\rightarrow+\infty}{\lim}||e_{i}(t)||\leq\varLambda_{i}\leq\varphi_{i8}(\overline{\iota}_{\psi},\overline{\iota}_{\omega},\mathcal{T}^{i}).$
			Also we have $\varLambda_{i}\leq\underset{i\in\{1,2,...,N\}}{\max}\varphi_{i8}(\overline{\iota}_{\psi},\overline{\iota}_{\omega},\mathcal{T}^{i})$.
		}This means that by decreasing the design parameters $\overline{\iota}_{\psi},\overline{\iota}_{\omega},\mathcal{T}^{i}$,
		the regulation error can be made arbitrary small.  Meanwhile, from
		the event-triggered condition (\ref{eq:26-1-1-1}), we can see when
		$\overline{\iota}_{\psi},\overline{\iota}_{\omega}$ are small, the
		condition (\ref{eq:26-1-1-1}) may be easier to be triggered, thus
		resulting in more frequent transmissions and higher communication
		burden. The detailed expression of $\varphi_{i8}(\overline{\iota}_{\psi},\overline{\iota}_{\omega},\mathcal{T}^{i})$
		is given by (D.29) in Appendix D. Note that this estimation
		may be a little conservative in some situations. This is a common
		phenomenon when using the Lyapunov function method (see the discussion
		in \cite{key-20}). However, the property of $\varphi_{i8}(\overline{\iota}_{\psi},\overline{\iota}_{\omega},\mathcal{T}^{i})$
		can give some insights on how the control parameters will influence
		the control accuracy. }
	
	\textcolor{black}{In addition, from (D.29) in Appendix D,
		we know when $\overline{\iota}_{\psi}=0,\overline{\iota}_{\omega}=0$
		and the signal $v$ is a constant, $\varphi_{i8}(\overline{\iota}_{\psi},\overline{\iota}_{\omega},\mathcal{T}^{i})=0$.
		This means we can make the regulation error converge to exact zero
		for constant $v$.}
\end{rem}
\begin{rem}
	According to Appendix D, the sampling period $\mathcal{T}^{i}$ should
	simultaneously satisfy (\ref{eq:16-1}) and
	\begin{equation}
	0<\frac{||R_{i}B_{i}K_{i}||\delta_{i4}(\mathcal{T}^{i})}{1-\delta_{i4}(\mathcal{T}^{i})}<1\label{eq:16-1-2-1}
	\end{equation}
	where $\delta_{i4}(\mathcal{T}^{i})$ is a $\mathcal{K}$-class function
	such that $\delta_{i4}(\mathcal{T}^{i})=\mathcal{T}^{i}||A_{i}+B_{i}K_{i}||\mathrm{e}^{||A_{i}||\mathcal{T}^{i}}$,
	$R_{i}$ is a given matrix such that $R_{i}(A_{i}+B_{i}K_{i})+(A_{i}+B_{i}K_{i})^{\mathrm{\mathit{T}}}R_{i}=-2I$.
	\textcolor{black}{Similar to Remarks \ref{rem:According-to-Appendix}
		and \ref{rem:From-Appendix-C,}, the minimum inter-event time can
		be determined by solving (\ref{eq:16-1}) and (\ref{eq:16-1-2-1})
		numerically.}
\end{rem}
\begin{rem}
	It is noted that we have assumed the sampling time $t_{k}$ for the
	communication between different agents have been synchronized as in
	\cite{key-3a-1,key-9-1}. This can be achieved by using some time
	synchronization methods such as \cite{key-b1,key-b4}. This is a common
	assumption in continuous time cooperative control. It should be also
	emphasized that the PET transmission between various agents, in sensor-to-controller
	and controller-to-actuator channels are all asynchronous. Also note
	that the selection of the sampling period $T$ could rely on some
	global graph information. One can choose a small enough sampling period
	by considering all the possible situations of the graph, or use some
	methods, $\emph{e.g.,}$ \cite{key-8a}, to estimate the graph information
	distributedly. In addition, from the simulations, we can see that
	the proposed method is robust to the variations of the sampling period.
\end{rem}

\section{\textcolor{black}{Simulations}}

Consider a linear multi-agent system described by (\ref{eq:1})-(\ref{eq:2})
with $4$ followers. The system matrix of the leader is $S=\mathrm{col}([0\thinspace1],[-1\thinspace0])$.
The dynamics of the followers are described as:
\begin{align*}
\dot{x}_{i} & =\left[\begin{array}{cc}
0 & 1\\
0 & \delta_{i}
\end{array}\right]x_{i}+\left[\begin{array}{c}
0\\
1
\end{array}\right]u_{i},\\
e_{i} & =[1\thinspace0]x_{i}+[-1\thinspace0]v,\\
y_{mi} & =[1\thinspace0]x_{i},i=1,2,3,4
\end{align*}
where $\delta_{1}=-0.3$,  $\delta_{2}=-0.4$, $\delta_{i3}=-0.5$,
$\delta_{i3}=-0.4$. The underlying communication graph is depicted
in Fig. \ref{fig:3-1-1}. The considered system can describe a class
of robotic systems \citep{key-14aa}. The initial states are $v(0)=[0.9\thinspace-0.5]^{\mathrm T}$,
$x_{1}(0)=[0.2\thinspace0.3]^{\mathrm T}$, $x_{2}(0)=[0.1\thinspace0.3]^{\mathrm T}$,
$x_{3}(0)=[0.5\thinspace0.6]^{\mathrm T}$, $x_{4}(0)=[0.8\thinspace0.8]^{\mathrm T}$.
The control purpose is to regulate $e_{i}$ to zero or a small neighborhood
around the origin, $i.e.$, make the output $y_{mi}$ of each follower
track the output $y_{0}\triangleq[-1\thinspace0]v$ of the leader.
The simulations are divided into two parts.

\subsection{PETM-A and PETM-B are invoked, PETM-C is not invoked}

We first consider the case when the controller-to-actuator channel
is continuous. 

\subsubsection{Effectiveness of the proposed method}

The PET distributed observer and output feedback controller are respectively
given by (\ref{eq:5-1-1})-(\ref{eq:26}) and (\ref{eq:10})-(\ref{eq:26-1-1})
with $\mu_{1}=\mu_{2}=3$, $\iota_{S}=\iota_{v}=2$, $\gamma_{S}=\gamma_{v}=1$;
$K_{i}=[-5\thinspace-5]$, $L_{i}=[-5\thinspace-5]$, $\sigma_{i}=\rho_{i}=1(i=1,2,3,4)$,
$\kappa=30$, $\iota_{\psi}=2$, $\gamma_{\psi}=1$. Based on (\ref{eq:10-1})
and (\ref{eq:16-1}), the sampling period is selected as $T=\mathcal{T}^{i}=0.01s(i=1,2,3,4)$.
The control performance of the multi-agent systems is shown in Fig.
\ref{fig:1}. It can be seen that the output of each follower quickly
follows the output of the leader. Meanwhile, the regulation errors
of the four followers all converge to zero. This indicates that cooperative
output regulation has been achieved. 

Fig. \ref{fig:2}(a) and Fig. \ref{fig:2}(b) show the event-triggered
instants for the communication between each agent pair, and the sensor-to-controller
transmission in each agent respectively. We observe many time intervals
which do not have data communication. This implies that the communication
burden has been reduced considerably. Fig. \ref{fig:3} shows the
inter-event time for the communication from agent 2 to agent 3, and
the sensor-to-controller transmission in agent 2. We can see that
a minimum inter-event time has been guaranteed. Moreover, all the
inter-event times are multiples of sampling period. These results
highlight the advantages of the PET output regulation.

\subsubsection{Discussions on the control parameters}

We will give some discussions on how the control parameters will influence
the control performance of the multi-agent systems. 

Fig. \ref{fig:1-2} shows the control performance of the multi-agent
systems for different $\mu_{1},\mu_{2}$. We can see that the convergence
speed increases a lot when the values of $\mu_{1},\mu_{2}$ change
from 1 to 2. However, the convergence rate for $\mu_{1}=\mu_{2}=3$
is slightly faster than that for $\mu_{1}=\mu_{2}=2$. This verifies
Remark \ref{rem:Note-that-theoretically,} and shows that when $\mu_{1},\mu_{2}$
are too large, the convergence speed does not increase considerably.
Therefore, $\mu_{1},\mu_{2}$ are recommended to be set between 2
and 3 for the simulation to save energy.

Tables \ref{tab:1}-\ref{tab:2} respectively show the event-triggered
times for different values of $\iota_{S},\iota_{v}\iota_{\psi}$ and
$\gamma_{S},\gamma_{v}\gamma_{\psi}$ in PETM-A and PETM-B. It can
be clearly seen that the communication load is reduced when increasing
$\iota_{S},\iota_{v}\iota_{\psi}$ and decreasing $\gamma_{S},\gamma_{v}\gamma_{\psi}$.
Fig. \ref{fig:1-2-2} demonstrates the control performance when $\iota_{S},\iota_{v}\iota_{\psi}$
and $\gamma_{S},\gamma_{v}\gamma_{\psi}$ take different values. We
can see that a longer convergence time is needed for larger $\iota_{S},\iota_{v}\iota_{\psi}$
and smaller $\gamma_{S},\gamma_{v}\gamma_{\psi}$. This verifies Remark
\ref{rem:For-the-parameters}.

Let us see how the parameters $\sigma_{i}$ and $\rho_{i}$ in (\ref{eq:11})
influence the performance of the multi-agent systems. Fig. \ref{fig:8}
shows the outputs of the four followers for different $\sigma_{i},\rho_{i}$.
We can see that the proposed controller is robust to the variations
of $\sigma_{i},\rho_{i}$. The control performance is almost the same
for different $\sigma_{i},\rho_{i}$. This implies that the selection
of $\sigma_{i},\rho_{i}$ can be very flexible. Table \ref{tab:3}
demonstrates the event-triggered times for different $\sigma_{i},\rho_{i}$.
We can see that as stated in Remark \ref{rem:For-the-parameters},
the communication burden is reduced considerably when $\sigma_{i}=\rho_{i}=1$.
This also shows the advantages of the control scheme in Fig. \ref{fig:3-1}
such that the controller should send some information to PETM-B to
further reduce the communication burden.

Finally, let us see how the selections of the sampling periods $T$
and $\mathcal{T}^{i}$ will influence the control performance. As
stated in Section 3.2, the sampling periods can be selected independently
and be different from one another. Fig. \ref{fig:9} shows the control
performance of the multi-agent systems under different $T,\mathcal{T}^{i}$.
We can see that the multi-agent systems are robust to the variations
of $T,\mathcal{T}^{i}$. 

\subsection{PETM-A, PETM-B and PETM-C are invoked}

We assume that PETM-C is invoked in Fig. \ref{fig:3-1}. The PET distributed
observer and output feedback controller are given by (\ref{eq:5-1-1})-(\ref{eq:26})
and (\ref{eq:10-2})-(\ref{eq:26-1-1-1}) with $\overline{\iota}_{\psi}=\overline{\iota}_{\omega}=0.001$,
$\iota_{\omega}=2$, $\gamma_{\omega}=1$. The other control parameters
are the same as before. The regulation errors and control efforts
are shown in Figs. \ref{fig:10}-\ref{fig:11}. It can be seen that
the regulation error has converged to a small neighborhood around
the origin. \textcolor{black}{According to (\ref{eq:phi8}) in Appendix
	D, we know for $\forall i\in\{1,2,...,N\}$,  $\varLambda_{i}\leq\underset{i\in\{1,2,...,N\}}{\max}\varphi_{i8}(\overline{\iota}_{\psi},\overline{\iota}_{\omega},\mathcal{T}^{i})\leq0.43$
	with $\varepsilon=0.01$, $\zeta_{1}=0.4$, $\zeta_{2}=1$, $\zeta_{3}=0.4$.} This is accord with Fig. \ref{fig:11}.
Note that as stated in Remark \ref{rem:11} the estimation of $\varLambda$
by (\ref{eq:phi8}) is somewhat conservative. Nevertheless, it can
provide guidelines for the selections of the control parameters. 

Fig. \ref{fig:10} shows the control performance under different $\overline{\iota}_{\psi},\overline{\iota}_{\omega}$.
We observe that the regulation errors all converge to a small neighborhood
of the origin for all the cases. This verifies the robustness of the
proposed method. Meanwhile, Table \ref{tab:5} shows the regulation
errors and event-triggered times for PETM-B and PETM-C with different
$\overline{\iota}_{\psi},\overline{\iota}_{\omega}$. We can see that
larger $\overline{\iota}_{\psi},\overline{\iota}_{\omega}$ can result
in larger regulation errors but small communication burden. This verifies
Remark \ref{rem:11}.

In addition, Fig. \ref{fig:11-1} demonstrates the control performance
when $S=\overline{\iota}_{\psi}=\overline{\iota}_{\omega}=0$. In
this case, $v$ is a constant signal. It can be seen that the regulation
error converges to exact zero even though PETM-C is invoked. This
also verifies Remark \ref{rem:11}.
\begin{figure}
	\begin{centering}
		\includegraphics[scale=0.8]{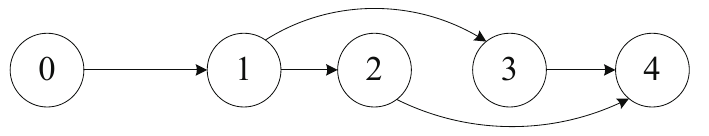}
		\par\end{centering}
	\caption{\label{fig:3-1-1}Communication graph.}
\end{figure}
\begin{figure}
	\begin{centering}
		\includegraphics[scale=0.55]{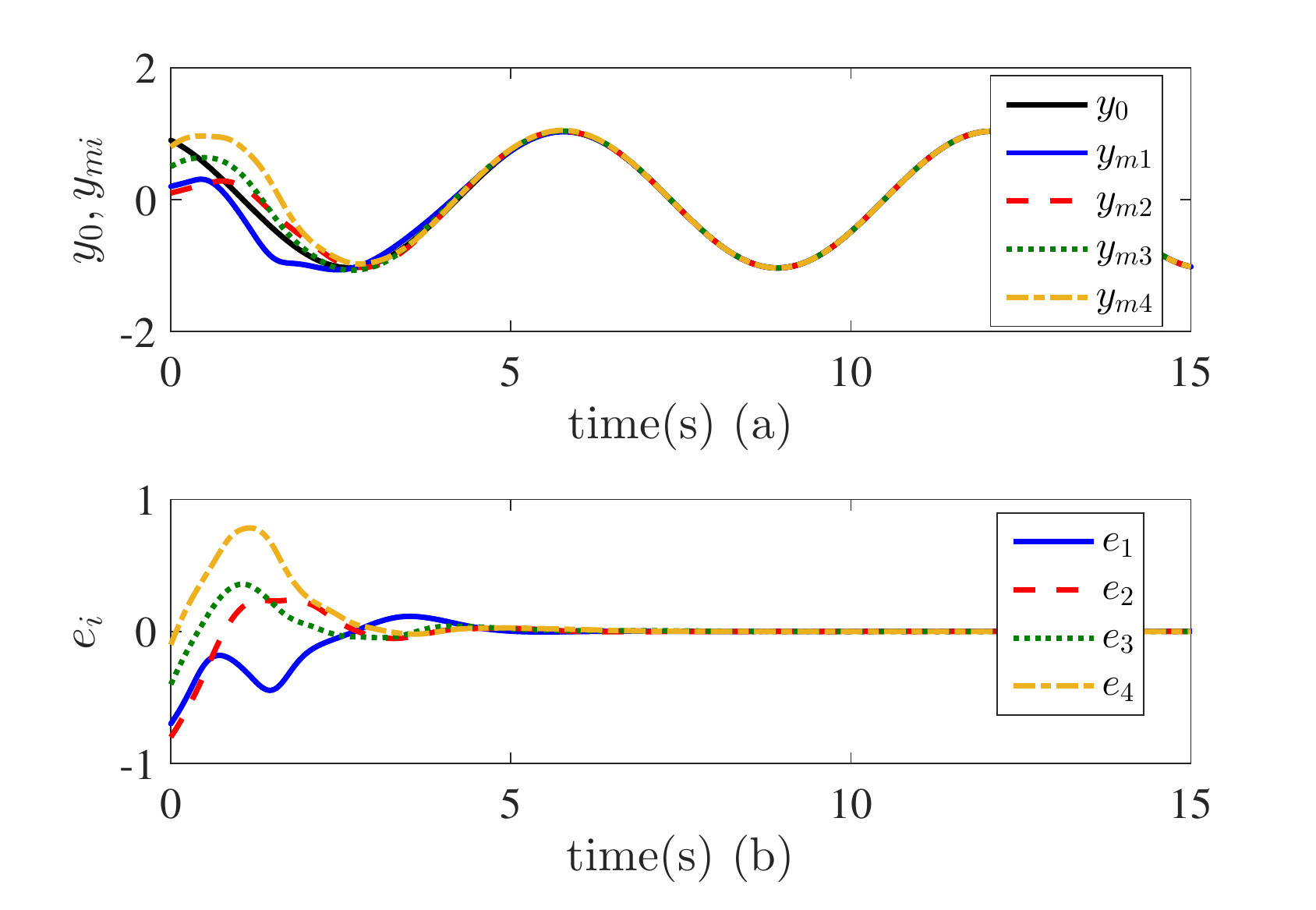}
		\par\end{centering}
	\caption{\label{fig:1}Control performance when PETM-C is not invoked. (a)
		Measurement outputs of leader and 4 followers; (b) Output regulation
		errors of 4 followers.}
\end{figure}
\begin{figure}
	\begin{centering}
		\includegraphics[scale=0.55]{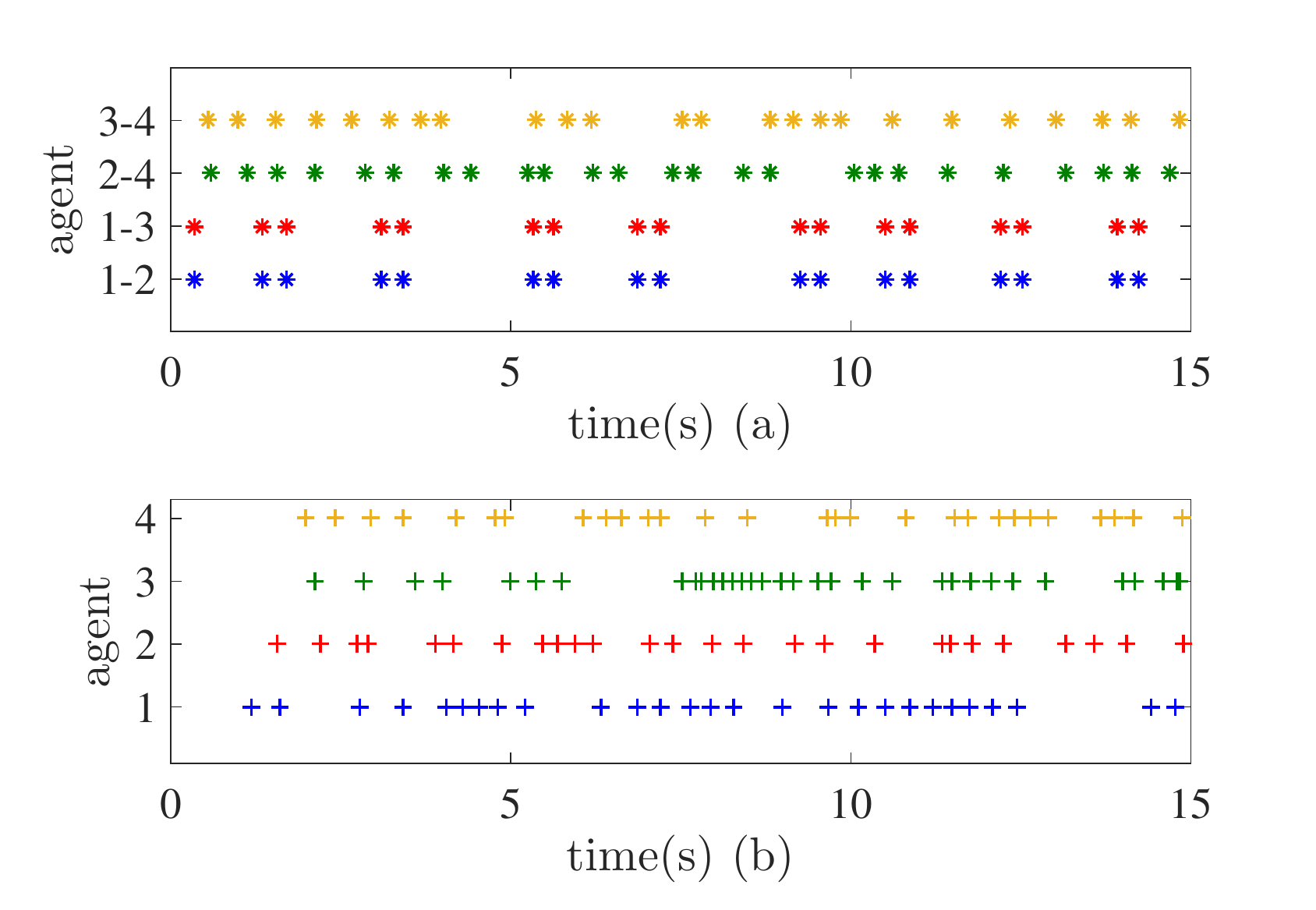}
		\par\end{centering}
	\caption{\label{fig:2}Event-triggered time instants. (a) Event-triggered time
		instants between each agent pair for PETM-A; (b) Event-triggered time
		instants of each agent for PETM-B. }
\end{figure}
\begin{figure}
	\begin{centering}
		\includegraphics[scale=0.55]{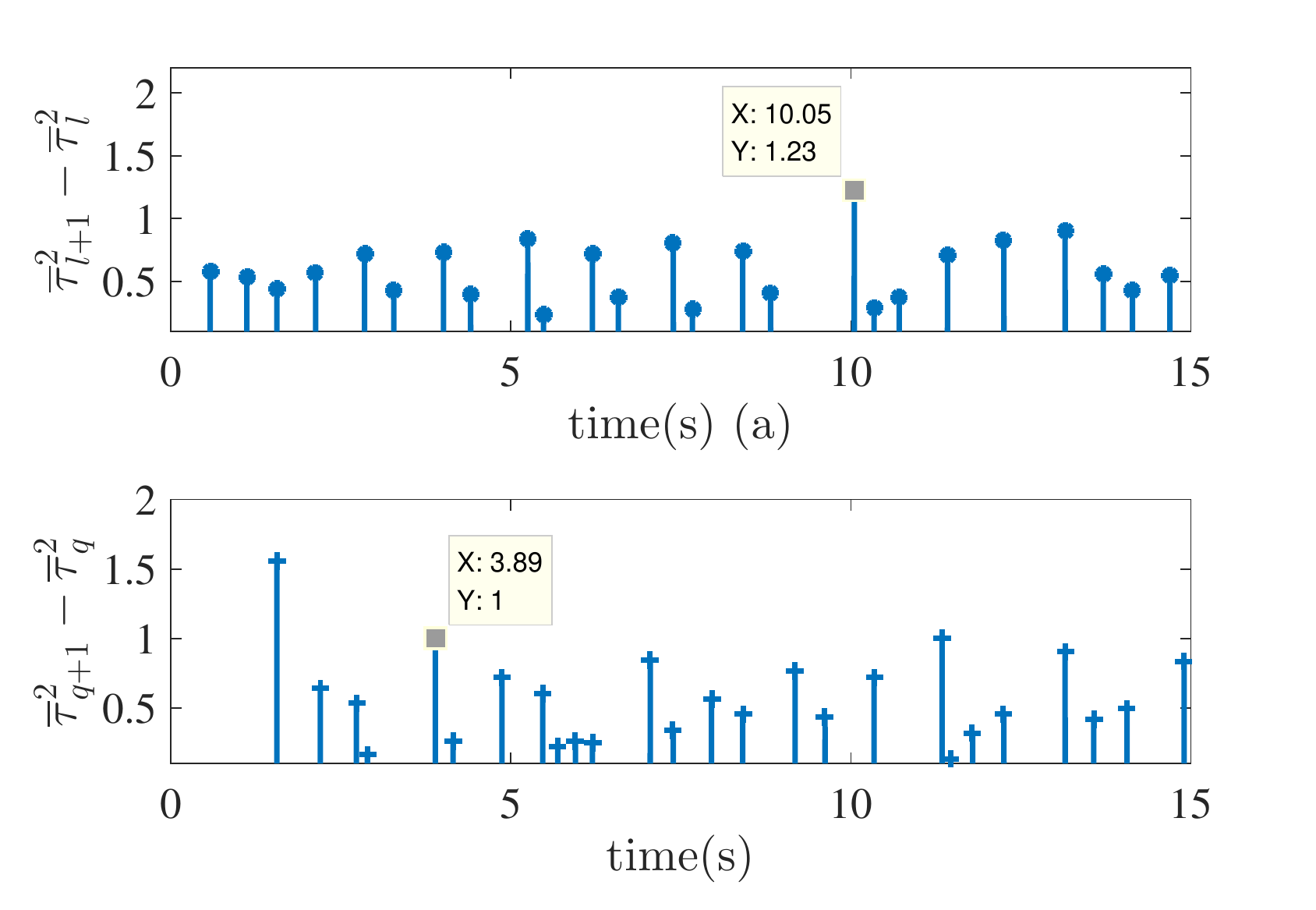}
		\par\end{centering}
	\caption{\label{fig:3}Inter-event time. (a) Inter-event time for the communication
		between agent 2 and 3; (b) Inter-event time between the sensor and
		controller in agent 2.}
\end{figure}
\begin{figure}
	\begin{centering}
		\includegraphics[scale=0.55]{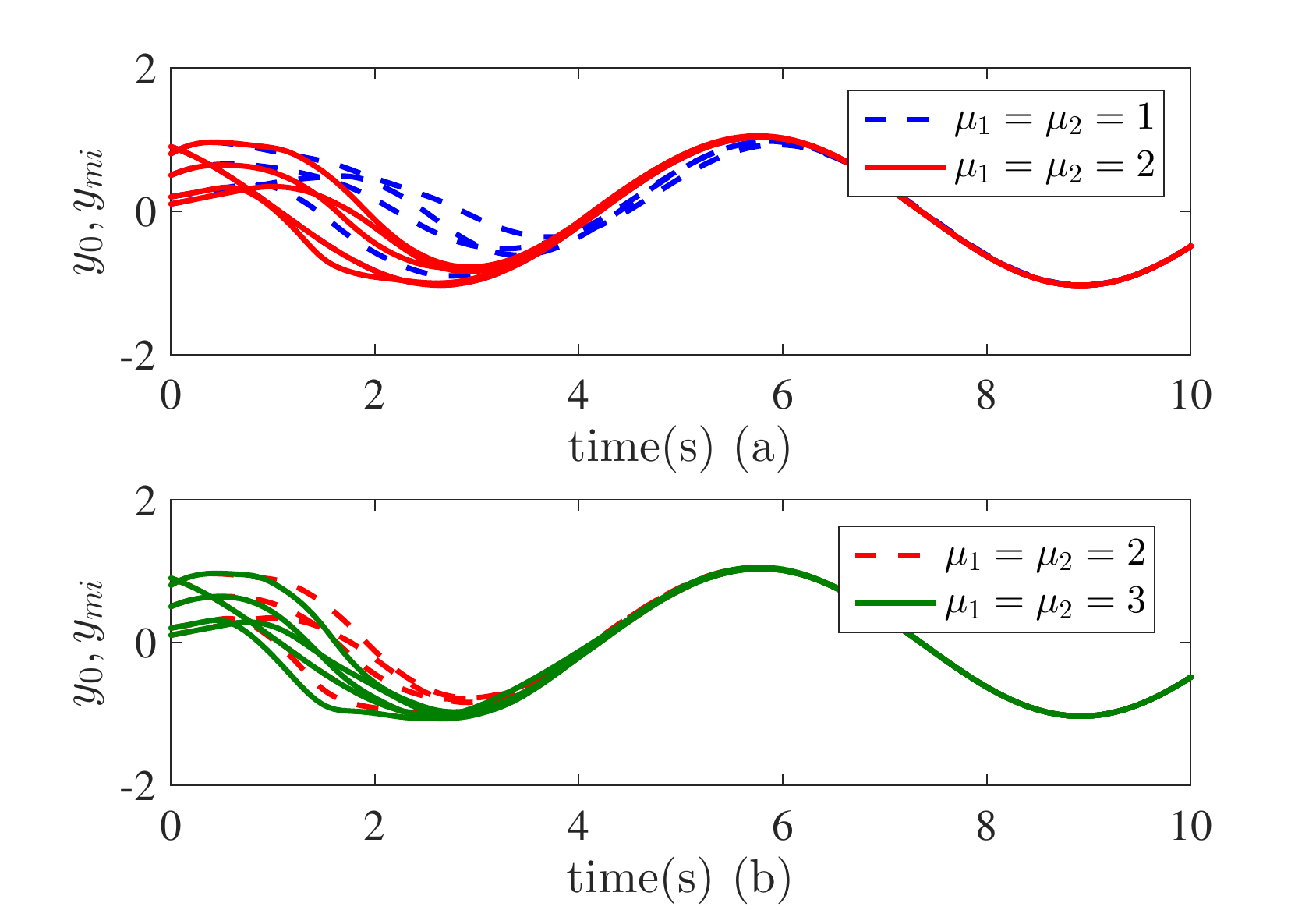}
		\par\end{centering}
	\caption{\label{fig:1-2}Control performance for different $\mu_{1},\mu_{2}$.}
\end{figure}
\begin{table}
	\centering{}\caption{\label{tab:1}Event-triggered times for PETM-A under different $\iota_{S},\iota_{v},\gamma_{S},\gamma_{v}$.}
	\begin{tabular}{ccccc}
		\hline 
		agent & 1-2 & 1-3 & 2-4 & 3-4\tabularnewline
		\hline 
		$\iota_{S}=\iota_{v}=2,\gamma_{S}=\gamma_{v}=1$ & 18 & 18 & 26 & 25\tabularnewline
		$\iota_{S}=\iota_{v}=1,\gamma_{S}=\gamma_{v}=2$ & 30 & 30 & 37 & 37\tabularnewline
		$\iota_{S}=\iota_{v}=0.5,\gamma_{S}=\gamma_{v}=2.5$ & 79 & 79 & 120 & 94\tabularnewline
		\hline 
	\end{tabular}
\end{table}
\begin{table}
	\centering{}\caption{\label{tab:2}Event-triggered times for PETM-B under different $\iota_{\psi},\gamma_{\psi}$.}
	\begin{tabular}{ccccc}
		\hline 
		agent & 1 & 2 & 3 & 4\tabularnewline
		\hline 
		$\iota_{\psi}=3,\gamma_{\psi}=0.5$ & 14 & 14 & 10 & 11\tabularnewline
		$\iota_{\psi}=2,\gamma_{\psi}=1$ & 27 & 26 & 20 & 24\tabularnewline
		$\iota_{\psi}=0.5,\gamma_{\psi}=1$ & 32 & 39 & 41 & 60\tabularnewline
		\hline 
	\end{tabular}
\end{table}
\begin{figure}
	\begin{centering}
		\includegraphics[scale=0.55]{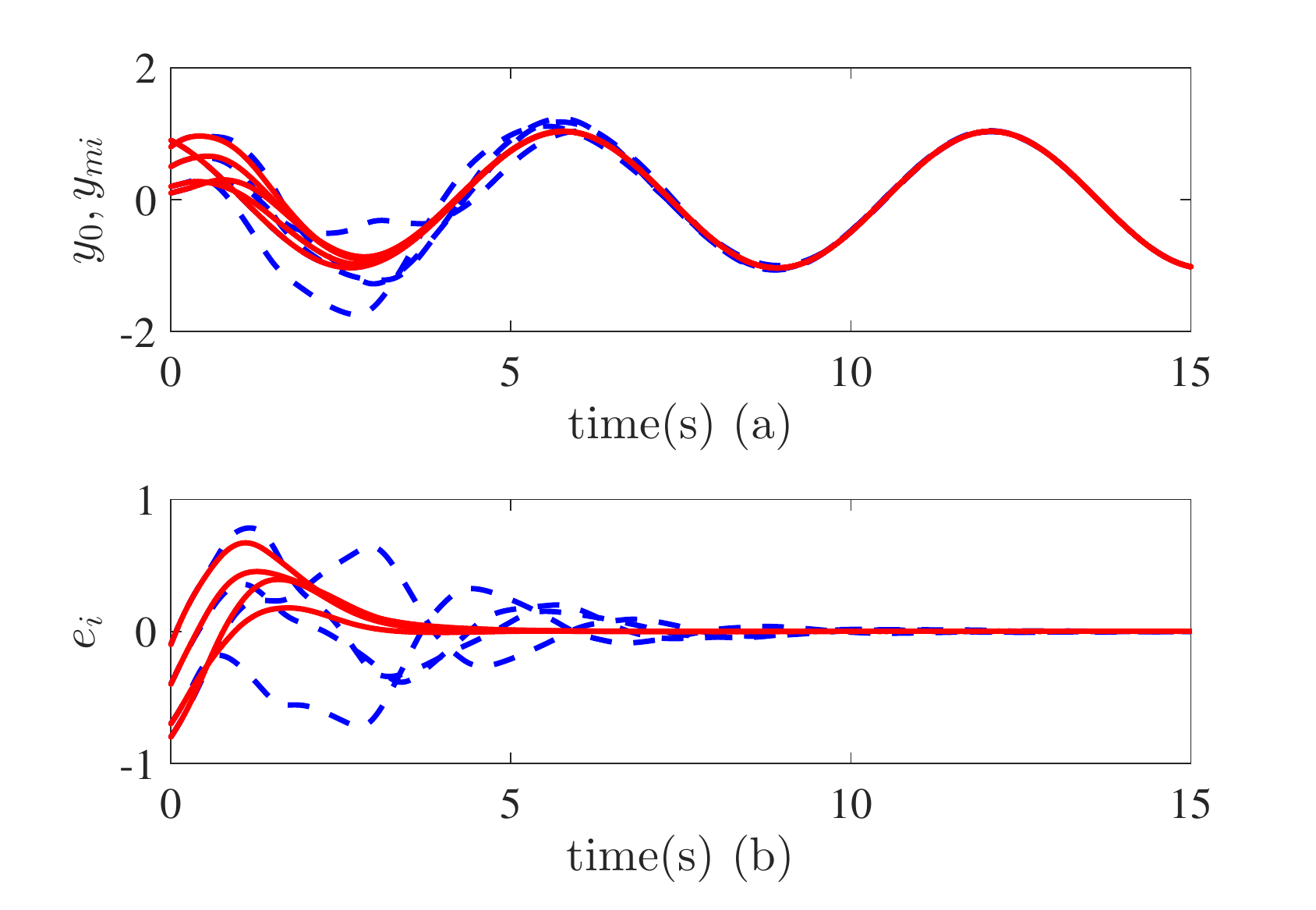}
		\par\end{centering}
	\caption{\label{fig:1-2-2}Control performance for different event-triggered
		conditions. (a) Measurement outputs of leader and 4 followers; (b)
		Output regulation errors of 4 followers. Blue dashed line: $\iota_{S}=\iota_{v}=2,\gamma_{S}=\gamma_{v}=1,\iota_{\psi}=3,\gamma_{\psi}=0.5$;
		red solid line: $\iota_{S}=\iota_{v}=0.5,\gamma_{S}=\gamma_{v}=2.5,\iota_{\psi}=0.5,\gamma_{\psi}=1$.}
\end{figure}
\begin{figure}
	\begin{centering}
		\includegraphics[scale=0.55]{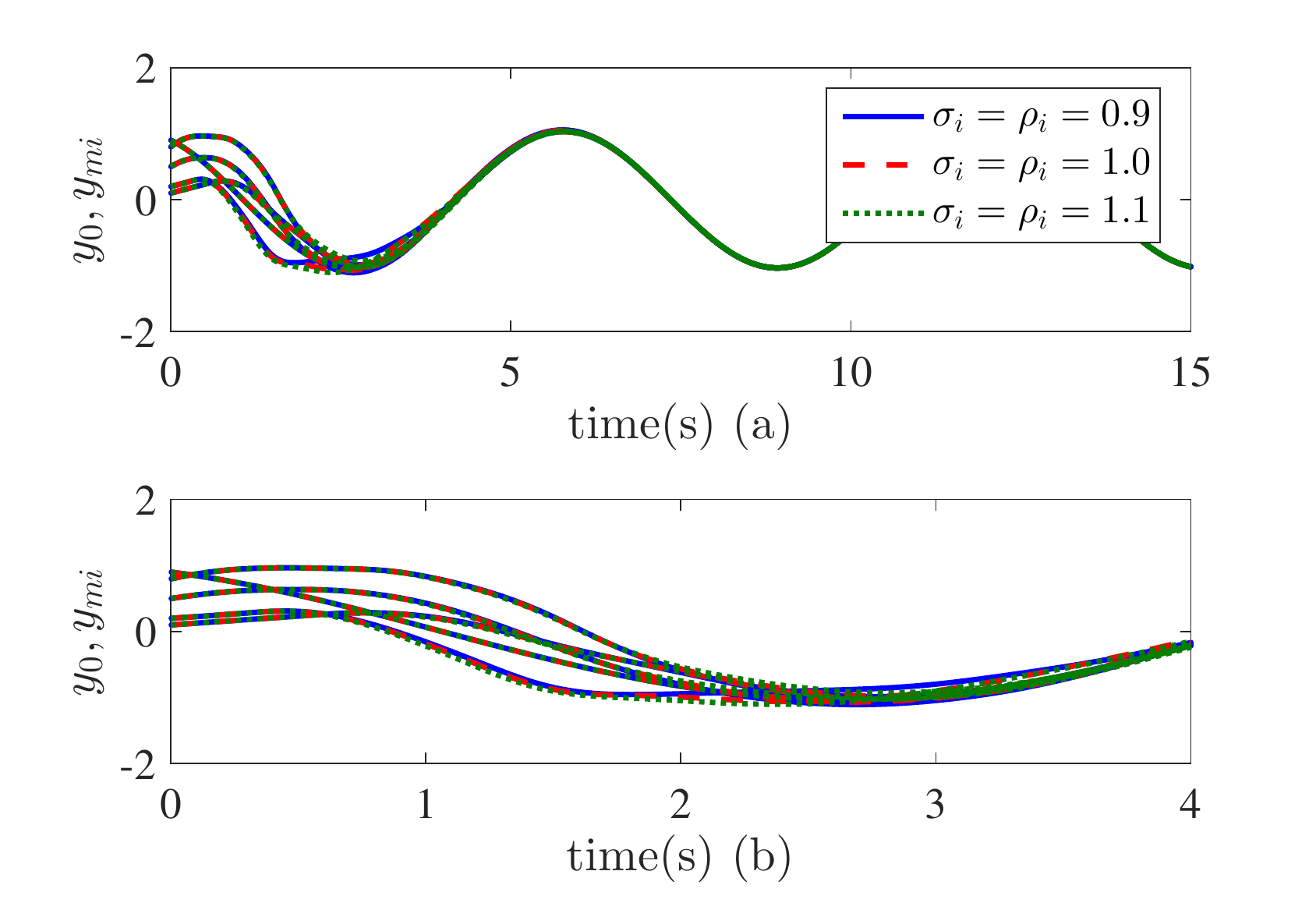}
		\par\end{centering}
	\caption{\label{fig:8}Control performance for different $\sigma_{i},\rho_{i}$.}
\end{figure}
\begin{table}
	\centering{}\caption{\label{tab:3}Event-triggered times for PETM-B under different $\sigma_{i},\rho_{i}$.}
	\begin{tabular}{ccccccc}
		\hline 
		$\sigma_{i},\rho_{i}$ & 0 & $0.8$ & $0.9$ & $1.0$ & $1.1$ & $1.2$\tabularnewline
		\hline 
		PETM-B & 971 & 845 &  769 & 34 & 765 & 843\tabularnewline
		\hline 
	\end{tabular}
\end{table}
\begin{figure}
	\begin{centering}
		\includegraphics[scale=0.55]{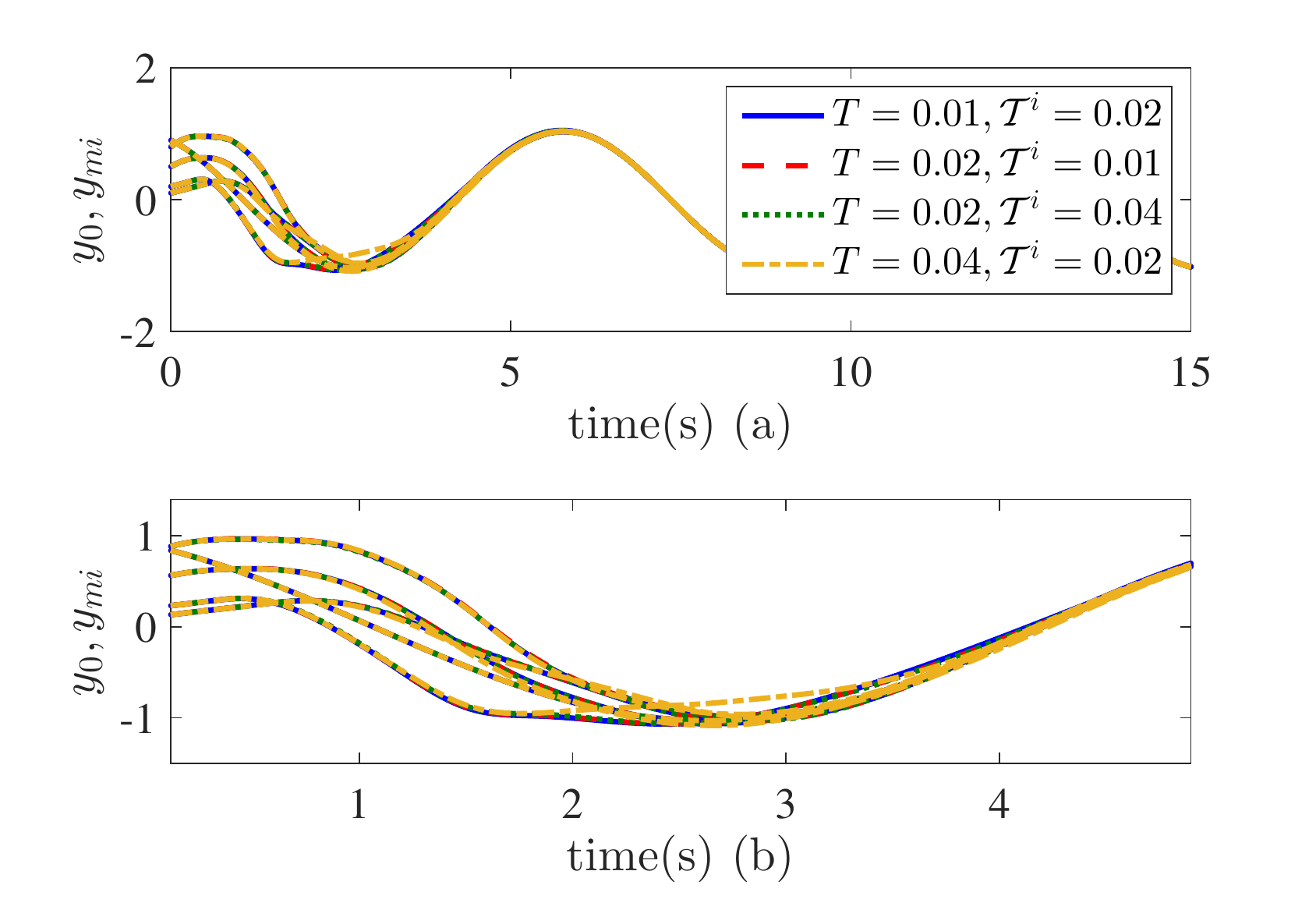}
		\par\end{centering}
	\caption{\label{fig:9}Control performance for different sampling periods $T,\mathcal{T}^{i}$.
		(a) Measurement outputs of leader and 4 followers; (b) Partial enlarged
		view of subfigure (a). }
\end{figure}
\begin{figure}
	\begin{centering}
		\includegraphics[scale=0.55]{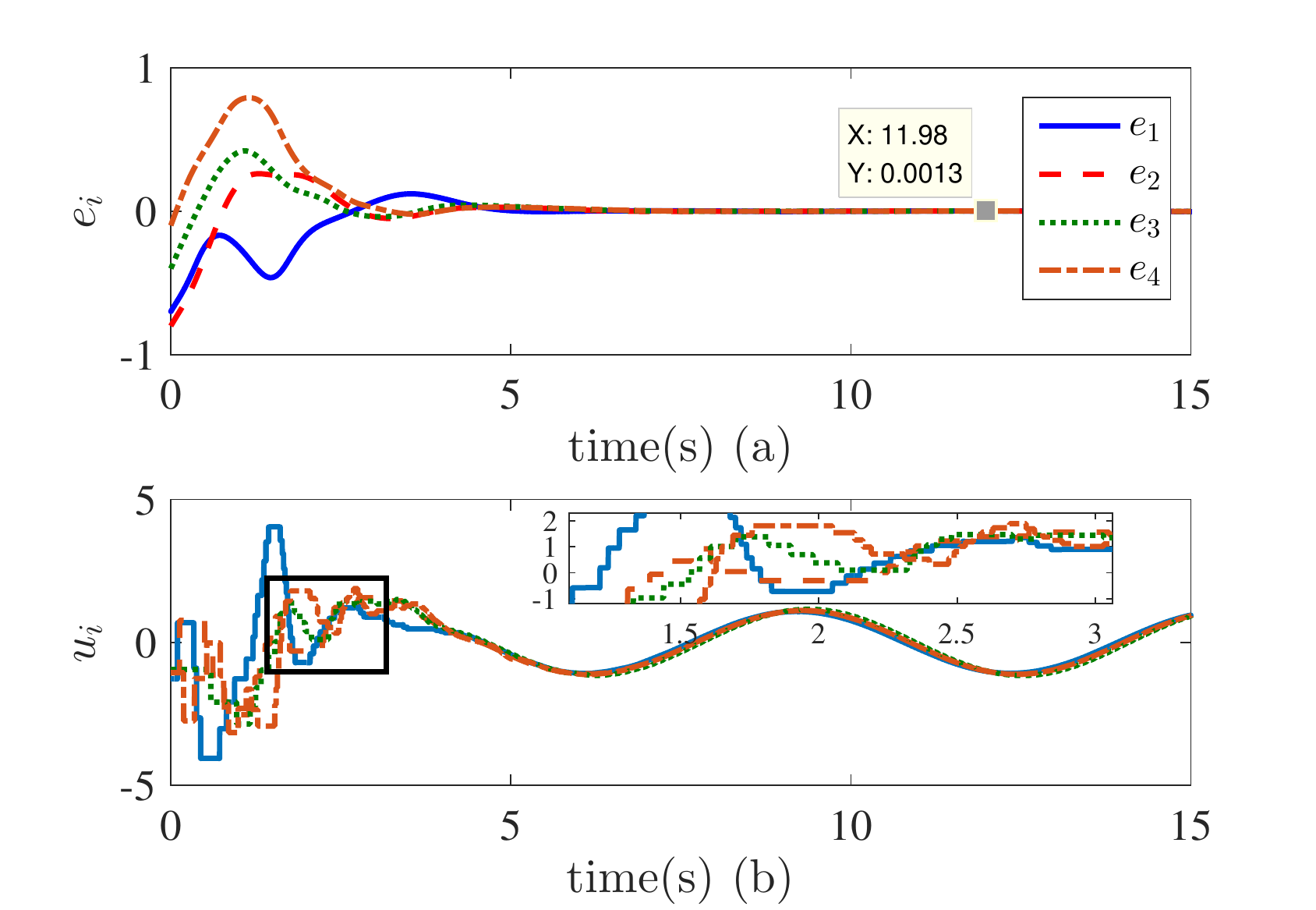}
		\par\end{centering}
	\caption{\label{fig:11}Control performance when PETM-C is invoked. (a) Output
		regulation errors of 4 followers; (b) Control efforts of 4 followers.}
\end{figure}
\begin{figure}
	\begin{centering}
		\includegraphics[scale=0.55]{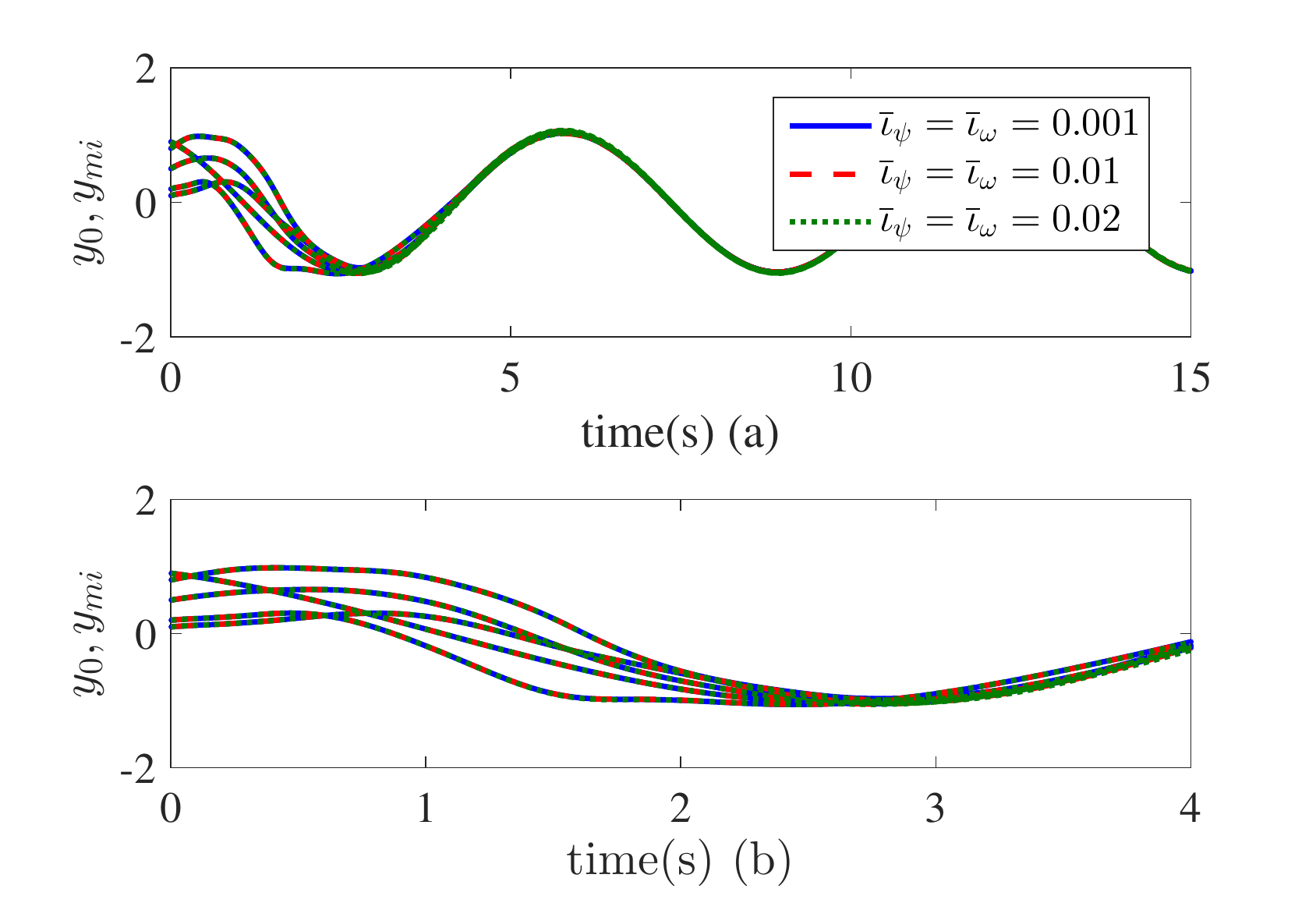}
		\par\end{centering}
	\caption{\label{fig:10}Control performance for different $\overline{\iota}_{\psi},\overline{\iota}_{\omega}$
		when PETM-C is invoked. (a) Measurement outputs of leader and 4 followers;
		(b) Partial enlarged view of subfigure (a).}
\end{figure}
\begin{table}
	\begin{centering}
		\caption{\label{tab:5}Event-triggered times for PETM-B under different $\overline{\iota}_{\psi},\overline{\iota}_{\omega}$.}
		\par\end{centering}
	\centering{}%
	\begin{tabular}{cccc}
		\hline 
		& error & PETM-B & PETM-C\tabularnewline
		\hline 
		$\overline{\iota}_{\psi}=\overline{\iota}_{\omega}=0$ & 0.0008 & 286 & 1015\tabularnewline
		$\overline{\iota}_{\psi}=\overline{\iota}_{\omega}=0.001$ & 0.0013 & 17 & 984\tabularnewline
		$\overline{\iota}_{\psi}=\overline{\iota}_{\omega}=0.01$ & 0.0070 & 13 & 597\tabularnewline
		$\overline{\iota}_{\psi}=\overline{\iota}_{\omega}=0.02$ & 0.0132 & 12 & 366\tabularnewline
		\hline 
	\end{tabular}
\end{table}
\begin{figure}
	\begin{centering}
		\includegraphics[scale=0.55]{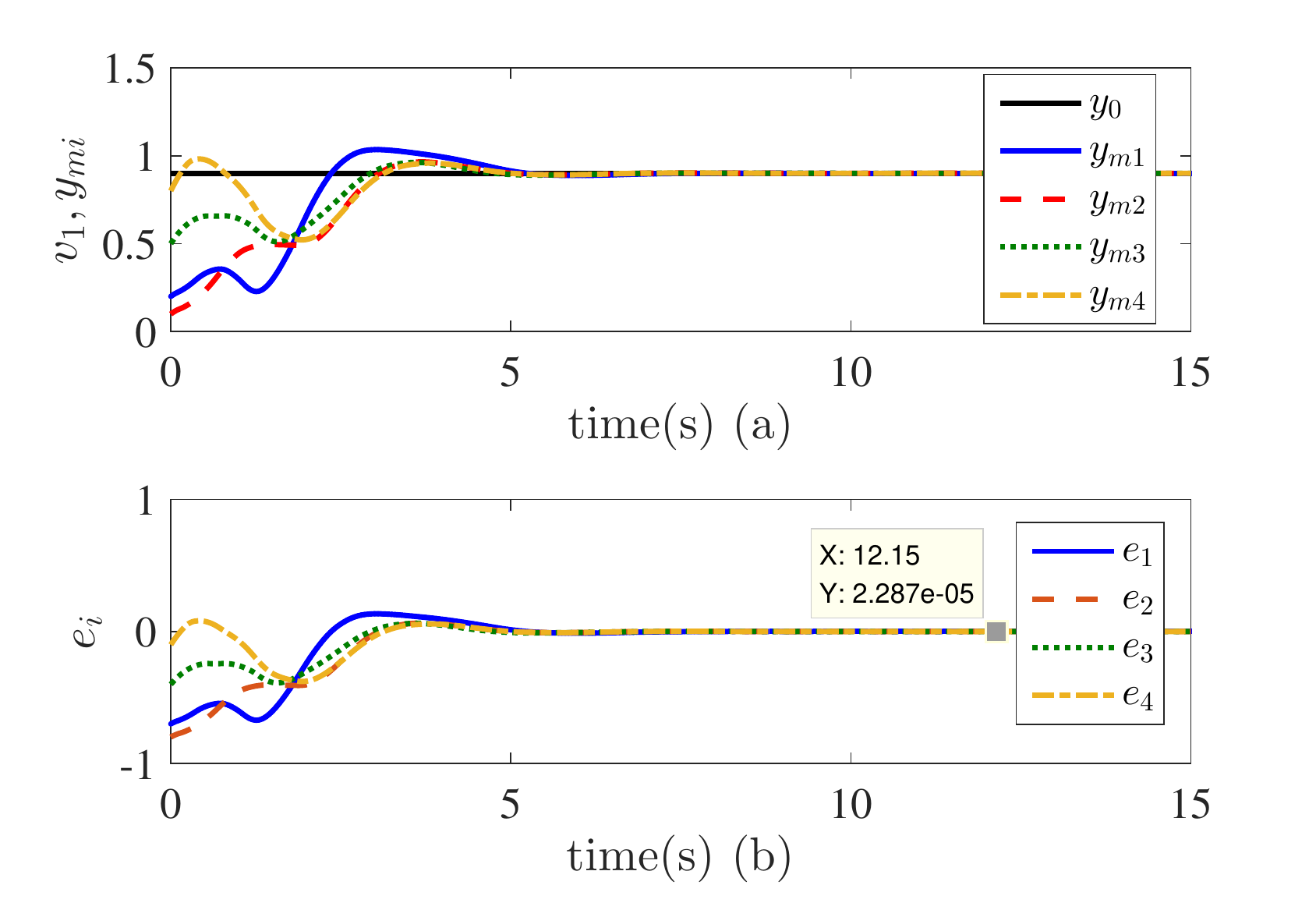}
		\par\end{centering}
	\caption{\label{fig:11-1}Control performance when PETM-C is invoked and $S=0,\overline{\iota}_{\psi}=\overline{\iota}_{\omega}=0$.
		(a) Measurement outputs of leader and 4 followers; (b) Output regulation
		errors of 4 followers.}
\end{figure}

\section{Conclusions}

In this paper, a new PET distributed observer and dynamic output feedback
controller are proposed to solve the cooperative output regulation
problem. PET mechanisms are simultaneously considered
in the communication among different agents, sensor-to-controller
and controller-to-actuator channels. Future works include considering
fully distributed PET output regulation problem under asynchronous
sampled data.


\appendix{}

\section{Proof of Theorem \ref{thm:1}}
\begin{proof}
	In the following let $c_{i}(i=1,2,...,16)$ denote some proper positive
	constants. $\varepsilon$ represents an arbitrary small positive constant.
	The proof is then divided into the following steps.
	
	\textit{Step 1}. We will find the relation between $||\hat{S}_{i}(t)-\hat{S}_{i}(t_{k})||$
	and $||\hat{S}_{i}(t)-S||(i=1,2,...,N)$. 
	
	Integrating (\ref{eq:5-1-1}) on time interval $[t_{k},t)\subseteq[t_{k},t_{k+1})$,
	we get
	\begin{equation}
	\hat{S}_{i}(t)-\hat{S}_{i}(t_{k})=(t-t_{k})\mu_{1}\sum_{j=0}^{N}a_{ij}(\hat{S}_{j}(\overline{t}_{l'}^{j})-\hat{S}_{j}(\overline{t}_{l}^{i})).\label{eq:9-1-1}
	\end{equation}
	It follows that
	\begin{align}
	\hat{S}-\hat{S}(t_{k})= & -\mu_{1}(t-t_{k})(\mathcal{H}\varotimes I)(\hat{S}(\overline{t}_{l})-\hat{S}(t_{k}))\nonumber \\
	& -\mu_{1}(t-t_{k})(\mathcal{H}\varotimes I)(\hat{S}(t_{k})-\hat{S})\nonumber \\
	& -\mu_{1}(t-t_{k})(\mathcal{H}\varotimes I)\tilde{S}\label{eq:32}
	\end{align}
	where $\hat{S}\triangleq\mathrm{col}(\hat{S}_{1},\hat{S}_{2},...,\hat{S}_{N})$,
	$\check{S}\triangleq\mathrm{col}(S,S,...,S)$, $\tilde{S}\triangleq\hat{S}-\check{S}$,
	$\hat{S}(\overline{t}_{l})\triangleq\mathrm{col}(\hat{S}_{1}(\overline{t}_{l'}^{1}),\hat{S}_{2}(\overline{t}_{l'}^{2}),...,\hat{S}_{N}(\overline{t}_{l'}^{N}))$.
	
	Next, we transform (\ref{eq:32}) into a vector form. Denote $\tilde{S}=[\tilde{s}_{1}\thinspace\tilde{s}_{2}\thinspace...\thinspace\tilde{s}_{n_{v}}]$,
	$\hat{S}(t_{k})-\hat{S}=[\Delta\hat{s}_{1}(t_{k})\thinspace\Delta\hat{s}_{2}(t_{k})\thinspace...\thinspace\Delta\hat{s}_{n_{v}}(t_{k})]$
	and $\hat{S}(\overline{t}_{l})-\hat{S}(t_{k})=[\Delta\hat{s}_{1}(\overline{t}_{l})\thinspace\Delta\hat{s}_{2}(\overline{t}_{l})\thinspace...\thinspace\Delta\hat{s}_{n_{v}}(\overline{t}_{l})]$
	where $\tilde{s}_{m},\Delta\hat{s}_{m}(t_{k}),\Delta\hat{s}_{m}(\overline{t}_{l})(m=1,2,...,n_{v})\in\mathbb{R}^{n_{v}N\times1}$
	are all column vectors. Then we have
	\begin{align*}
	\Delta\hat{s}_{m}(t_{k}) & =-\mu_{1}(t-t_{k})(\mathcal{H}\varotimes I)(\Delta\hat{s}_{m}(\overline{t}_{l})+\Delta\hat{s}_{m}(t_{k})-\tilde{s}_{m}).
	\end{align*}
	
	It follows that
	\begin{align*}
	||\Delta\hat{s}_{m}(t_{k})|| & \leq\delta_{1}(T)(||\Delta\hat{s}_{m}(\overline{t}_{l})||+||\Delta\hat{s}_{m}(t_{k})||+||\tilde{s}_{m}||)
	\end{align*}
	where $\delta_{1}(T)$ is a $\mathcal{K}$-class function such that
	\begin{equation}
	\delta_{1}(T)=\mu_{1}T||\mathcal{H}\varotimes I||.\label{eq:17}
	\end{equation}
	
	Therefore, there exists a small enough $T$ such that $\delta_{1}(T)<1$.
	Then we obtain
	\begin{align*}
	||\Delta\hat{s}_{m}(t_{k})|| & \leq\delta_{S}(T)||\tilde{s}_{m}||+\delta_{S}(T)||\Delta\hat{s}_{m}(\overline{t}_{l})||
	\end{align*}
	where 
	\begin{equation}
	\delta_{S}(T)=\frac{\delta_{1}(T)}{1-\delta_{1}(T)}.\label{eq:18}
	\end{equation}
	Finally, according to the event-triggered condition (\ref{eq:26}),
	we know
	\begin{align}
	||\Delta\hat{s}_{m}(t_{k})|| & \leq\delta_{S}(T)||\tilde{s}_{m}||+c_{1}\mathrm{e}^{-\gamma_{S}t_{k}}.\label{eq:31-1}
	\end{align}
	
	\textit{Step 2}. We will show $\tilde{S}_{i}$ converges to zero exponentially. 
	
	Note that (\ref{eq:5-1-1}) can be transformed into
	\begin{align*}
	\dot{\tilde{s}}_{m} & =-\mu_{1}(\mathcal{H}\varotimes I)(\tilde{s}_{m}+\Delta\hat{s}_{m}(\overline{t}_{l})+\Delta\hat{s}_{m}(t_{k}))
	\end{align*}
	for $m=1,2,...,n_{v}$.
	
	Take the following Lyapunov function 
	\[
	V_{m}=\frac{1}{2}\tilde{s}_{m}^{T}(P\varotimes I)\tilde{s}_{m}
	\]
	where $P$ is a positive matrix such that $P\mathcal{H}+\mathcal{H}^{T}P=2I$.
	Note that $P$ exists due to $-\mathcal{H}$ is Hurwitz from Lemma
	\ref{lem:Under-Assumption-,}.
	
	Then, we get
	\begin{align*}
	\dot{V}_{m}= & -\mu_{1}||\tilde{s}_{m}||^{2}-\mu_{1}\tilde{s}_{m}^{T}(P\mathcal{H}\varotimes I)(\Delta\hat{s}_{m}(\overline{t}_{l})+\Delta\hat{s}_{m}(t_{k})).
	\end{align*}
	
	Using Young's inequality and (\ref{eq:26}), (\ref{eq:31-1}), we
	obtain
	\begin{align}
	\dot{V}_{m}\leq & -\left(\mu_{1}-\mu_{1}||P\mathcal{H}\varotimes I||\delta_{S}(T)-\varepsilon\right)||\tilde{s}_{m}||^{2}\nonumber \\
	& +c_{2}\mathrm{e}^{-2\gamma_{S}t_{k}}.\label{eq:34}
	\end{align}
	
	Therefore, when $\delta_{S}(T)$ satisfies 
	\begin{equation}
	\delta_{S}(T)<\frac{1}{||P\mathcal{H}\varotimes I||},\label{eq:21-2}
	\end{equation}
	we have $\mu_{1}-\mu_{1}||P\mathcal{H}\varotimes I||\delta_{S}(T)-\varepsilon>0$.
	Then solving equation (\ref{eq:34}), we can obtain that $V_{m},\tilde{s}_{m},\tilde{S}$
	will all converge to zero exponentially. 
	
	Combining (\ref{eq:17}), (\ref{eq:18}) and (\ref{eq:21-2}), note
	$||P\mathcal{H}\varotimes I||=||P\mathcal{H}||$, $||\mathcal{H}\varotimes I||=||\mathcal{H}||$,
	we know $T$ should satisfy
	\begin{equation}
	T<\frac{1}{\mu_{1}\left(||P\mathcal{H}||+1\right)||\mathcal{H}||}.\label{eq:22-1}
	\end{equation}
	
	\textit{Step 3}. We will show $\hat{v}$ does not exhibit finite time
	escape, $i.e.$, $\hat{v}$ is bounded on a finite time interval $[0,t)$. 
	
	Let $\overline{v}_{i}^{*}(t,t_{k},\overline{t}_{l}^{i})=\mathrm{e}^{\hat{S}_{i}(\overline{t}_{l}^{i})(t-t_{k})}\hat{v}_{i}(t_{k})(i=1,2,...,N)$,
	then (\ref{eq:6-1}) can be written as:
	\begin{align}
	\dot{\hat{v}}_{i}= & \hat{S}_{i}(\overline{t}_{l}^{i})\hat{v}_{i}(t)+\mu_{2}\sum_{j=0}^{N}a_{ij}(\overline{v}_{j}^{*}(t,t_{k},\overline{t}_{l'}^{j})-\overline{v}_{i}^{*}(t,t_{k},\overline{t}_{l}^{i}))\nonumber \\
	& +\mu_{2}\sum_{j=0}^{N}a_{ij}(\Delta\upsilon_{j}(t,t_{k},\overline{t}_{l'}^{j})-\Delta\upsilon_{i}(t,t_{k},\overline{t}_{l}^{i}))\label{eq:32-1}
	\end{align}
	where $\Delta\upsilon_{i}(t,t_{k},\overline{t}_{l}^{i})=\overline{v}_{i}(t,\overline{t}_{l}^{i})-\overline{v}_{i}^{*}(t,t_{k},\overline{t}_{l}^{i}).$
	
	Solving (\ref{eq:32-1}) on time interval $[t_{k},t_{k+1})$, we have
	\begin{align}
	\hat{v}_{i}(t)= & \overline{v}_{i}^{*}(t,t_{k},\overline{t}_{l}^{i})-\mu_{2}(t-t_{k})\sum_{j=0}^{N}a_{ij}\overline{v}_{i}^{*}(t,t_{k},\overline{t}_{l}^{i})\nonumber \\
	& +\sum_{j=0}^{N}\mu_{2}\int_{t_{k}}^{t}a_{ij}\mathrm{e}^{-\hat{S}_{i}(\overline{t}_{l}^{i})(\tau-t)}\overline{v}_{j}^{*}(\tau,t_{k},\overline{t}_{l'}^{j})d\tau\label{eq:22}\\
	& +\Delta\overline{\upsilon}_{i}\nonumber 
	\end{align}
	where 
	\begin{align*}
	\Delta\overline{\upsilon}_{i}= & \int_{t_{k}}^{t}\mathrm{e}^{-\hat{S}_{i}(\overline{t}_{l}^{i})(\tau-t)}\left(\mu_{2}\sum_{j=0}^{N}a_{ij}(\Delta\upsilon_{j}(\tau)-\Delta\upsilon_{i}(\tau))\right)d\tau.
	\end{align*}
	
	For the terms in (\ref{eq:22}), we have
	\begin{align}
	& \mu_{2}\int_{t_{k}}^{t}a_{ij}\mathrm{e}^{-\hat{S}_{i}(\overline{t}_{l}^{i})(\tau-t)}\overline{v}_{j}^{*}(\tau,t_{k},\overline{t}_{l'}^{j})d\tau\nonumber \\
	= & \mu_{2}(t-t_{k})a_{ij}\overline{v}_{j}^{*}(t,t_{k},\overline{t}_{l'}^{j})+\mu_{2}a_{ij}E_{ij}(t)\overline{v}_{j}^{*}(t,t_{k},\overline{t}_{l'}^{j})\label{eq:11-1-1}
	\end{align}
	where $E_{ij}(t)$ is a time-varying matrix such that
	\begin{equation}
	E_{ij}(t)\triangleq\int_{t_{k}}^{t}(\mathrm{e}^{-\hat{S}_{i}(\overline{t}_{l}^{i})(\tau-t)}-\mathrm{e}^{-\hat{S}_{j}(\overline{t}_{l'}^{j})(\tau-t)})\mathrm{e}^{\hat{S}_{j}(\overline{t}_{l'}^{j})(\tau-t)}d\tau.\label{eq:12-1}
	\end{equation}
	
	Substituting (\ref{eq:11-1-1}) into (\ref{eq:22}), we have
	\begin{align*}
	\hat{v}_{i}(t)= & \overline{v}_{i}^{*}(t,t_{k},\overline{t}_{l}^{i})\\
	& +\mu_{2}(t-t_{k})\sum_{j=0}^{N}a_{ij}(\overline{v}_{j}^{*}(t,t_{k},\overline{t}_{l'}^{j})-\overline{v}_{i}^{*}(t,t_{k},\overline{t}_{l}^{i}))\\
	& +\mu_{2}\sum_{j=0}^{N}a_{ij}E_{ij}(t)\overline{v}_{j}^{*}(t,t_{k},\overline{t}_{l'}^{j})\\
	& +\Delta\overline{\upsilon}_{i}.
	\end{align*}
	For $i=1,2,...,N$, we have
	\begin{align}
	& \hat{v}-\overline{v}^{*}(t,t_{k},\overline{t}_{l})\nonumber \\
	= & \mu_{2}(t-t_{k})(\mathcal{H}\varotimes I)(\overline{v}^{*}(t,t_{k},\overline{t}_{l})-\check{v})+\mu_{2}\epsilon+\Delta\overline{\upsilon}\label{eq:38}
	\end{align}
	where $\hat{v}\triangleq\mathrm{col}(\hat{v}_{1},\hat{v}_{2},...,\hat{v}_{N})$,
	$\check{v}\triangleq\mathrm{col}(v,v,...,v)$, 
	\[
	\overline{v}^{*}(t,t_{k},\overline{t}_{l})\triangleq\mathrm{col}(\overline{v}_{1}^{*}(t,t_{k},\overline{t}_{l'}^{1}),\overline{v}_{2}^{*}(t,t_{k},\overline{t}_{l'}^{2}),...,\overline{v}_{N}^{*}(t,t_{k},\overline{t}_{l'}^{N})),
	\]
	\[
	\epsilon\triangleq\mathrm{col}\left(\epsilon_{1},\epsilon_{2}...,\epsilon_{N}\right),
	\]
	\[
	\Delta\overline{\upsilon}\triangleq\mathrm{col}\left(\Delta\overline{\upsilon}_{1},\Delta\overline{\upsilon}_{2},...,\Delta\overline{\upsilon}_{N}\right)
	\]
	with $\epsilon_{i}=\sum_{j=0}^{N}a_{ij}E_{ij}(t)\overline{v}_{j}^{*}(t,t_{k},\overline{t}_{l'}^{j})(i=1,2,...,N)$.
	
	Note that for the elements in $\epsilon$ and using Lemma \ref{lem:(Matrix-norm-inequalities)},
	\begin{align}
	||a_{ij}E_{ij}(t)\overline{v}_{j}^{*}(t,t_{k},\overline{t}_{l'}^{j})||\leq a_{ij}||\overline{v}_{j}^{*}||\cdot||E_{ij}(t)||.\label{eq:16-2}
	\end{align}
	
	From (\ref{eq:12-1}) and Appendix B, we know $||E_{ij}(t)||$ converges
	to zero exponentially. Hence, using (\ref{eq:16-2}) for $\epsilon$,
	there exist positive constants $c_{3},c_{4}$ such that
	\begin{align}
	||\epsilon||\leq & c_{3}\mathrm{e}^{-c_{4}t}||\overline{v}^{*}(t,t_{k},\overline{t}_{l})||.\label{eq:21}
	\end{align}
	Meanwhile, according to Appendix B, we know there exist positive constants
	$c_{5},c_{6}$ such that
	\begin{align}
	||\Delta\overline{\upsilon}||\leq & c_{5}\mathrm{e}^{-c_{6}t}.\label{eq:21-1}
	\end{align}
	
	Then using (\ref{eq:21}) and (\ref{eq:21-1}) for (\ref{eq:38}),
	we obtain:
	\begin{align}
	||\hat{v}-\overline{v}^{*}(t,t_{k},\overline{t}_{l})||\leq & (\delta_{2}(T)+\mu_{2}c_{3}\mathrm{e}^{-c_{4}t})||\overline{v}^{*}(t,t_{k},\overline{t}_{l})||\nonumber \\
	& +\delta_{2}(T)||\check{v}||+c_{5}\mathrm{e}^{-c_{6}t}\label{eq:40-1}
	\end{align}
	where $\delta_{2}(T)$ is a $\mathcal{K}$-class function such that
	\begin{equation}
	\delta_{2}(T)=\mu_{2}T||\mathcal{H}\varotimes I||.\label{eq:31-2}
	\end{equation}
	It follows that
	\begin{align*}
	||\hat{v}|| & \leq c_{7}||\overline{v}^{*}(t,t_{k},\overline{t}_{l})||+c_{8}\mathrm{e}^{-c_{9}t}+c_{10}\\
	& \mathrm{\leq\mathit{c}_{7}e}^{||\hat{S}_{i}(\overline{t}_{l}^{i})(t-t_{k})||}\hat{v}_{i}(t_{k})+c_{8}\mathrm{e}^{-c_{9}t}+c_{10}
	\end{align*}
	It can be seen that $\hat{v}$ will not exhibit finite time escape.
	
	\textit{Step 4}. We will find the relation between $||\hat{v}-\overline{v}^{*}(t,t_{k},\overline{t}_{l})||$
	and $||\tilde{v}||\triangleq||\hat{v}-\check{v}||$. Note that (\ref{eq:38})
	can be expressed as:
	\begin{align*}
	\hat{v}-\overline{v}^{*}(t,t_{k},\overline{t}_{l})= & \mu_{2}(t-t_{k})(\mathcal{H}\varotimes I)(\overline{v}^{*}(t,t_{k},\overline{t}_{l})-\hat{v})\\
	& +\mu_{2}(t-t_{k})(\mathcal{H}\varotimes I)\tilde{v}+\mu_{2}\epsilon+\Delta\overline{\upsilon}
	\end{align*}
	where $\tilde{v}\triangleq\hat{v}-\check{v}$.
	
	Using (\ref{eq:21}) and (\ref{eq:21-1}), we have
	\begin{align}
	& ||\hat{v}-\overline{v}^{*}(t,t_{k},\overline{t}_{l})||\nonumber \\
	\leq & (\delta_{2}(T)+\mu_{2}c_{3}\mathrm{e}^{-\mu_{1}c_{4}t})||\overline{v}^{*}(t,t_{k},\overline{t}_{l})-\hat{v}||\nonumber \\
	& +(\delta_{2}(T)+\mu_{2}c_{3}\mathrm{e}^{-c_{4}t})||\tilde{v}||\nonumber \\
	& +\mu_{2}c_{3}\mathrm{e}^{-c_{4}t}||\check{v}||+c_{5}\mathrm{e}^{-c_{6}t}.\label{eq:27}
	\end{align}
	
	Next, define a set $\Omega_{L}$ for the sampling instant $t_{k}$
	which will be used in the following analysis.
	
	By (\ref{eq:40-1}), we know if $T$ is small enough, $\delta_{2}(T)<1$
	is true. Thus, it can be concluded that there exists a finite integer
	$k_{1}$ such that for $\forall t\in[t_{k},t_{k+1})$ with $k\geq k_{1}$,
	we have $\delta_{2}(T)+\mu_{2}c_{3}\mathrm{e}^{-c_{4}t}<1$ and $\mu_{2}c_{3}\mathrm{e}^{-c_{4}t}<\varepsilon$
	where $\varepsilon$ can be an arbitrary small constant. 
	
	Also note that according to Step 2, we know $\hat{S}$ converges to
	$S$ exponentially. Therefore, there exist positive constants $c_{11},c_{12}$
	such that $||S_{d}||\leq c_{11}\mathrm{e}^{-c_{12}t}$ where $S_{d}(\overline{t}_{l})=\mathrm{diag}(\hat{S}_{1}(\overline{t}_{l'}^{1})-S,\hat{S}_{1}(\overline{t}_{l'}^{2})-S,...,\hat{S}_{N}(\overline{t}_{l'}^{N})-S)$.
	This means that there exists a finite integer $k_{2}$ such that for
	$\forall t\in[t_{k},t_{k+1})$ with $k\geq k_{2}$, we have $||S_{d}||\leq c_{11}\mathrm{e}^{-c_{12}t}<\varepsilon$
	where $\varepsilon$ is an arbitrary small constant. 
	
	Then define the set $\Omega_{L}$ as 
	\[
	\Omega_{L}\triangleq\{t_{k}|k\in\mathbb{N},k\geq\mathrm{max}(k_{1},k_{2})\}.
	\]
	Therefore, for $\forall t\in[t_{k},t_{k+1})$ with $t_{k},t_{k+1}\in\Omega_{L}$,
	(\ref{eq:27}) can be written as:
	\begin{align}
	||\hat{v}-\overline{v}^{*}(t,t_{k},\overline{t}_{l})||\leq\delta_{v}(T)||\tilde{v}||+c_{13}\mathrm{e}^{-c_{14}t}\label{eq:26-2}
	\end{align}
	where 
	\begin{equation}
	\delta_{v}(T)=\frac{(\delta_{2}(T)+\mu_{2}c_{3}\mathrm{e}^{-\mu_{1}c_{4}t})}{1-(\delta_{2}(T)+\mu_{2}c_{3}\mathrm{e}^{-\mu_{1}c_{4}t})}.\label{eq:34-1}
	\end{equation}
	
	\textit{Step 5}. We will show $\tilde{v}_{i}$ converges to zero exponentially. 
	
	We consider our analysis for $t\in[t_{k},t_{k+1})$ with $t_{k},t_{k+1}\in\Omega_{L}$.
	Note that from Step 3), we know $\hat{v}_{i}$ will not exhibit finite
	time escape. Hence, for any finite $k$, $\hat{v}_{i}$ is bounded
	on $[t_{k},t_{k+1})$ .
	
	By (\ref{eq:6-1}), we have
	\begin{align}
	\dot{\tilde{v}}= & (I\varotimes S-\mu_{2}\mathcal{H}\varotimes I)\tilde{v}+S_{d}\tilde{v}+S_{d}\check{v}\nonumber \\
	& -\mu_{2}(\mathcal{H}\varotimes I)(\overline{v}^{*}(t,t_{k},\overline{t}_{l})-\hat{v}).\label{eq:17-1}
	\end{align}
	
	Take the following Lyapunov function 
	\[
	V_{v}=\frac{1}{2}\tilde{v}^{T}(P\varotimes I)\tilde{v}
	\]
	where $P$ is a positive matrix such that $P\mathcal{H}+\mathcal{H}^{T}P=2I$.
	
	The derivative of $V_{v}$ is computed as:
	\begin{align*}
	\dot{V}_{v}= & \tilde{v}^{T}(P\varotimes I)(I\varotimes S-\mu_{2}\mathcal{H}\varotimes I)\tilde{v}\\
	& +\tilde{v}^{T}(P\varotimes I)S_{d}\tilde{v}+\tilde{v}^{T}(D\varotimes I)S_{d}\check{v}\\
	& -\tilde{v}^{T}(P\varotimes I)\mu_{2}(\mathcal{H}\varotimes I)(\overline{v}^{*}(t,t_{k},\overline{t}_{l})-\hat{v}).
	\end{align*}
	Using (\ref{eq:26-2}), we have
	\begin{align*}
	\dot{V}_{v}\leq & \tilde{v}^{T}(P\varotimes S)\tilde{v}-\mu_{2}\lambda||\tilde{v}||^{2}\\
	& +||\tilde{v}||^{2}||P\varotimes I||\cdot||S_{d}||+||\tilde{v}||\cdot||(P\varotimes I)S_{d}\check{v}||\\
	& +\mu_{2}||P\mathcal{H}\varotimes I||\cdot||\tilde{v}||\left(\delta_{v}(T)||\tilde{v}||+c_{13}\mathrm{e}^{-c_{14}t}\right).
	\end{align*}
	Note that due to $S$ is a skew-symmetric matrix, $(P\varotimes S)^{\rm T}=P\varotimes S^{\rm T}=-(P\varotimes S)$.
	This means that $P\varotimes S$ is also a skew-symmetric matrix. Hence,
	$\tilde{v}^{\rm T}(P\varotimes S)\tilde{v}=0$. Also note that $||S_{d}||\leq c_{11}\mathrm{e}^{-c_{12}t}<\varepsilon$
	for $t\in[t_{k},t_{k+1})$ with $t_{k},t_{k+1}\in\Omega_{L}$. Thus,
	by Young's inequality, we obtain:
	\begin{align*}
	\dot{V}_{v}\leq & -\left(\mu_{2}-||P\varotimes I||\varepsilon-\mu_{2}||P\mathcal{H}\varotimes I||\delta_{v}(T)-\varepsilon\right)||\tilde{v}||^{2}\\
	& +c_{15}\mathrm{e}^{-2c_{12}t}+c_{16}\mathrm{e}^{-2c_{14}t}.
	\end{align*}
	
	Hence, when 
	\begin{equation}
	\delta_{v}(T)<\mu_{2}\label{eq:36}
	\end{equation}
	we can obtain $\mu_{2}-||D\varotimes I||\varepsilon-\mu_{2}||D\mathcal{H}\varotimes I||\delta_{v}(T)-\varepsilon<0$.
	Then by solving the above equation, we can show $V_{v},\tilde{v}$
	converge to zero exponentially. The proof is completed.
	
	For the selection of $T$. Based on (\ref{eq:31-2}), (\ref{eq:34-1}),
	(\ref{eq:36}), (\ref{eq:22-1}) and note $\mu_{2}c_{3}\mathrm{e}^{-c_{4}t}<\varepsilon$,
	$c_{11}\mathrm{e}^{-c_{12}t}<\varepsilon$, we know $T$ should satisfy
	(\ref{eq:10-1}). 
\end{proof}

\section{Proposition \ref{prop:-converge-to} }
\begin{prop}
	\label{prop:-converge-to} $||E_{ij}(t)||_{F}$ in (\ref{eq:12-1})
	and $||\Delta\overline{\upsilon}||$ in (\ref{eq:38}) both converge
	to zero exponentially.
\end{prop}
\begin{proof}
	We first show $||E_{ij}(t)||_{F}$ converge to zero exponentially.
	\textcolor{black}{Note that from Steps 1-2 in Appendix A, we know $\tilde{S}_{i}$ converges to zero exponentially for $\forall i=1,2,...,N$. Hence, $\hat{S}_{i}(t_{k}),\hat{S}_{j}(t_{k})$ are both bounded.}
	Then by Lemma \ref{lem:(Matrix-norm-inequalities)}, we have
	\begin{align*}
	& ||E_{ij}(t)||\\
	\leq & \int_{t_{k}}^{t}||(\mathrm{e}^{-\hat{S}_{j}(t_{k})(\tau-t)}-\mathrm{e}^{-\hat{S}_{i}(t_{k})(\tau-t)})||\cdot||\mathrm{e}^{\hat{S}_{j}(t_{k})(\tau-t_{k})}||d\tau\\
	\leq & g_{1}\int_{t_{k}}^{t}||X(t_{k})(\tau-t)||\mathrm{e}^{||X(t_{k})(\tau-t)||+||Y(t_{k})(\tau-t)||}d\tau
	\end{align*}
	where $X(t)=\hat{S}_{j}(t)-\hat{S}_{i}(t)$, $Y(t)=\hat{S}_{i}(t)$,
	$g_{1}>0$ is a positive constant.
	
	Meanwhile, due to $||\tilde{S}||$ converge to zero exponentially,
	there exist constants $g_{2},g_{3}$ such that $||X||\leq g_{2}\mathrm{e}^{-g_{3}t}$.
	Hence, we have $||E_{ij}(t)||\leq g_{4}\mathrm{e}^{-g_{5}t}$where
	$g_{4},g_{5}>0$ are positive constants.
	
	Next, we shall show $||\Delta\overline{\upsilon}||$ converge to zero
	exponentially. Recalling (\ref{eq:26}) and the definitions of $\overline{v}_{i}(t,\overline{t}_{l}^{i})$,
	$\overline{v}_{i}^{*}(t,t_{k},\overline{t}_{l}^{i})$, we have
	\begin{align}
	||\Delta\upsilon_{i}(t,t_{k},\overline{t}_{l}^{i})|| & =||\overline{v}_{i}(t,\overline{t}_{l}^{i})-\overline{v}_{i}^{*}(t,t_{k},\overline{t}_{l}^{i})||\nonumber \\
	& =||\mathrm{e}^{\hat{S}_{i}(\overline{t}_{l}^{i})(t-\overline{t}_{l}^{i})}\hat{v}_{i}(\overline{t}_{l}^{i})-\mathrm{e}^{\hat{S}_{i}(\overline{t}_{l}^{i})(t-t_{k})}\hat{v}_{i}(t_{k})||\nonumber \\
	& \leq||\mathrm{e}^{\hat{S}_{i}(\overline{t}_{l}^{i})(t-t_{k})}||\cdot||\mathrm{e}^{\hat{S}_{i}(\overline{t}_{l}^{i})(t_{k}-\overline{t}_{l}^{i})}-\hat{v}_{i}(t_{k})||\nonumber \\
	& \leq g_{6}\mathrm{e}^{-g_{7}t}\label{eq:31}
	\end{align}
	where $g_{6},g_{7}>0$ are positive constants. Using the above equation
	for each element in $||\Delta\overline{\upsilon}||,$ the proof is
	completed.
\end{proof}

\section{Proof of Theorem \ref{thm:1-1}}
\begin{proof}
	The proof is divided into the following steps.
	
	\textit{Step 1).} We will find the relation between $\tilde{x}_{i}(t)$
	and $\tilde{x}_{i}(\tau_{p}^{i})$ where $\tilde{x}_{i}\triangleq\hat{x}_{i}-x_{i}$.
	
	Subtracting (\ref{eq:2}) from (\ref{eq:11}), we have
	\begin{align}
	\dot{\tilde{x}}_{i}= & A_{i}\tilde{x}_{i}+E_{i}\tilde{v}_{i}\nonumber \\
	& +L_{i}(C_{mi}\hat{x}_{i}(\tau_{p}^{i})+D_{mi}u_{i}(\tau_{p}^{i}))\nonumber \\
	& +L_{i}(\rho_{i}F_{i}\overline{v}_{i}(t,\overline{t}_{l}^{i})+(1-\sigma_{i})F_{mi}\overline{v}_{i}(t,\overline{t}_{l}^{i})+\psi_{i}(\tau_{p}^{i}))\nonumber \\
	& +L_{i}(\psi_{i}(\overline{\tau}_{q}^{i})-\psi_{i}(\tau_{p}^{i})).\label{eq:46}
	\end{align}
	Then using (\ref{eq:12}) for $\psi_{i}(\tau_{p}^{i})$, we obtain
	\begin{align}
	\dot{\tilde{x}}_{i}= & A_{i}\tilde{x}_{i}+E_{i}\tilde{v}_{i}+L_{i}(C_{mi}\hat{x}_{i}(\tau_{p}^{i})+D_{mi}u_{i}(\tau_{p}^{i}))\nonumber \\
	& +L_{i}(F_{mi}\overline{v}_{i}(\tau_{p}^{i},\overline{t}_{l}^{i})-y_{m}(\tau_{p}^{i}))\nonumber \\
	& +L_{i}(\psi_{i}(\overline{\tau}_{q}^{i})-\psi_{i}(\tau_{p}^{i})).\label{eq:47}
	\end{align}
	Adopting (\ref{eq:2}) for $y_{m}(\tau_{p}^{i})$, we have
	\begin{align}
	\dot{\tilde{x}}_{i}= & (A_{i}+L_{i}C_{mi})\tilde{x}_{i}+L_{i}C_{mi}\left(\tilde{x}_{i}(\tau_{p}^{i})-\tilde{x}_{i}\right)\nonumber \\
	& +E_{i}\tilde{v}_{i}+L_{i}F_{mi}\left(\overline{v}_{i}(\tau_{p}^{i},\overline{t}_{l}^{i})-v(\tau_{p}^{i})\right)\nonumber \\
	& +L_{i}\left(\psi_{i}(\overline{\tau}_{q}^{i})-\psi_{i}(\tau_{p}^{i})\right).\label{eq:14}
	\end{align}
	Note that
	\begin{align*}
	\overline{v}_{i}(\tau_{p}^{i},\overline{t}_{l}^{i})-v(\tau_{p}^{i})= & (\overline{v}_{i}(\tau_{p}^{i},\overline{t}_{l}^{i})-\hat{v}_{i}(\tau_{p}^{i})+\hat{v}_{i}(\tau_{p}^{i})-v(\tau_{p}^{i})).
	\end{align*}
	From Theorem \ref{thm:1}, we know $\overline{v}_{i}(\tau_{p}^{i},\overline{t}_{l}^{i})-\hat{v}_{i}(\tau_{p}^{i})$
	and $\hat{v}_{i}(\tau_{p}^{i})-v(\tau_{p}^{i})$ will both converge
	to zero exponentially. Hence, $\overline{v}_{i}(\tau_{p}^{i},\overline{t}_{l}^{i})-v(\tau_{p}^{i})$
	will converge to zero exponentially. Meanwhile, from Theorem \ref{thm:1}
	and (\ref{eq:26-1-1}) we know $\tilde{v}_{i}$ and $\psi_{i}(\overline{\tau}_{q}^{i})-\psi_{i}(\tau_{p}^{i})$
	will also converge to zero exponentially. Hence, we can conclude that
	there exist positive constants $\mathrm{c}_{1},\mathrm{c}_{2}$ such
	that 
	\begin{align}
	||E_{i}\tilde{v}_{i}+L_{i}F_{mi}\left(\overline{v}_{i}(\tau_{p}^{i},\overline{t}_{l}^{i})-v(\tau_{p}^{i})\right)\nonumber \\
	+L_{i}\left(\psi_{i}(\overline{\tau}_{q}^{i})-\psi_{i}(\tau_{p}^{i})\right)|| & <\mathrm{c}_{1}\mathrm{e}^{-\mathrm{c}_{2}t}.\label{eq:15}
	\end{align}
	Hence, integrating (\ref{eq:14}) on $[\tau_{p}^{i},\tau_{p+1}^{i})$
	we get
	\begin{align}
	\tilde{x}_{i}(t)-\tilde{x}_{i}(\tau_{p}^{i})= & \int_{\tau_{p}^{i}}^{t}A_{i}\left(\tilde{x}_{i}(\tau)-\tilde{x}_{i}(\tau_{p}^{i})\right)d\tau\nonumber \\
	& +(t-\tau_{p}^{i})(A_{i}+L_{i}C_{mi})\tilde{x}_{i}(\tau_{p}^{i})\nonumber \\
	& +\int_{\tau_{p}^{i}}^{t}\mathrm{c}_{1}\mathrm{e}^{-\mathrm{c}_{2}\tau}d\tau.\label{eq:50}
	\end{align}
	It follows that
	\begin{align}
	||\tilde{x}_{i}(t)-\tilde{x}_{i}(\tau_{p}^{i})||\leq & \int_{\tau_{p}^{i}}^{t}||A_{i}||\cdot||\tilde{x}_{i}(\tau)-\tilde{x}_{i}(\tau_{p}^{i})||d\tau\nonumber \\
	& +\mathcal{T}^{i}||A_{i}+L_{i}C_{mi}||\cdot||\tilde{x}_{i}(\tau_{p}^{i})||\nonumber \\
	& +\mathcal{T}^{i}\mathrm{c}_{1}\mathrm{e}^{-\mathrm{c}_{2}\tau_{p}^{i}}.\label{eq:51}
	\end{align}
	By Gronwall's inequality, we obtain
	\begin{align}
	||\tilde{x}_{i}(t)-\tilde{x}_{i}(\tau_{p}^{i})||\leq & \delta_{i3}(\mathcal{T}^{i})||\tilde{x}_{i}(\tau_{p}^{i})||\nonumber \\
	& +\mathcal{T}^{i}\mathrm{e}^{||A_{i}||\mathcal{T}^{i}}\mathrm{c}_{1}\mathrm{e}^{-\mathrm{c}_{2}\tau_{p}^{i}}\label{eq:42}
	\end{align}
	where $\delta_{i3}(\mathcal{T}^{i})$ is a $\mathcal{K}$-class function
	such that 
	\[
	\delta_{i3}(\mathcal{T}^{i})=\mathcal{T}^{i}||A_{i}+L_{i}C_{mi}||\mathrm{e}^{||A_{i}||\mathcal{T}^{i}}.
	\]
	When $\mathcal{T}^{i}$ is small enough such that $\delta_{i3}(\mathcal{T}^{i})<1$,
	then we have
	\begin{align}
	||\tilde{x}_{i}-\tilde{x}_{i}(\tau_{p}^{i})|| & \leq\delta_{\tilde{x}}(\mathcal{T}^{i})||\tilde{x}_{i}||+\mathrm{c}_{3}\mathrm{e}^{-\mathrm{c}_{2}t}\label{eq:16-1-1}
	\end{align}
	where $\mathrm{c}_{3}>0$ is a constant,
	\begin{equation}
	\delta_{\tilde{x}}(\mathcal{T}^{i})=\frac{\delta_{i3}(\mathcal{T}^{i})}{1-\delta_{i3}(\mathcal{T}^{i})}.\label{eq:34-1-1}
	\end{equation}
	
	\textit{Step 2). }We will show $\tilde{x}_{i},e_{i}$ converge to
	zero exponentially. 
	
	First, we demonstrate $\tilde{x}_{i}$ converge to zero exponentially.
	Note that $A_{i}+L_{i}C_{mi}$ is a Hurwitz matrix, then there exists
	a positive matrix $Q_{i}$ such that $Q_{i}(A_{i}+L_{i}C_{mi})+(A_{i}+L_{i}C_{mi})^{\mathrm{\mathit{T}}}Q_{i}=-2I$.
	Then consider the following Lyapunov function
	\[
	V_{\tilde{x}}=\frac{1}{2}\tilde{x}_{i}^{T}Q_{i}\tilde{x}_{i}.
	\]
	From (\ref{eq:14}), the derivative of $V_{\tilde{x}}$ is given by:
	\begin{align}
	\dot{V}_{\tilde{x}}= & -||\tilde{x}_{i}||^{2}+\tilde{x}_{i}^{T}Q_{i}L_{i}C_{mi}\left(\tilde{x}_{i}(\tau_{p}^{i})-\tilde{x}_{i}\right)\nonumber \\
	& +\tilde{x}_{i}^{T}Q_{i}\left(E_{i}\tilde{v}_{i}+L_{i}F_{mi}\left(\overline{v}_{i}(\tau_{p}^{i},\overline{t}_{l}^{i})-v(\tau_{p}^{i})\right)\right)\nonumber \\
	& +\tilde{x}_{i}^{T}Q_{i}L_{i}(\psi_{i}(\overline{\tau}_{q}^{i})-\psi_{i}(\tau_{p}^{i})).\label{eq:55}
	\end{align}
	Using (\ref{eq:42}) and (\ref{eq:15}), we get
	\begin{align*}
	\dot{V}_{\tilde{x}}\leq & -(1-||Q_{i}L_{i}C_{mi}||\delta_{\tilde{x}}(\mathcal{T}^{i})-\varepsilon)||\tilde{x}_{i}||^{2}+\mathrm{c}_{4}\mathrm{e}^{-2\mathrm{c}_{2}t}
	\end{align*}
	where $\varepsilon$ is an arbitrary small constant, $\mathrm{c}_{4}$
	is a positive constant.
	
	Therefore, when 
	\[
	||Q_{i}L_{i}C_{mi}||\delta_{\tilde{x}}(\mathcal{T}^{i})<1
	\]
	then we can conclude that $\tilde{x}_{i}$ will converge to zero exponentially.
	
	Next, we will prove $e_{i}$ converge to zero exponentially. Consider
	the following coordinate transformation $\overline{x}_{i}=x_{i}-X_{i}v$
	and $\overline{u}_{i}=u_{i}-U_{i}v$. Then based on (\ref{eq:4-1}),
	(\ref{eq:2}) is expressed as:
	\begin{align}
	\dot{\overline{x}}_{i}= & A_{i}(\overline{x}_{i}+X_{i}v)+B_{i}(\overline{u}_{i}+U_{i}v)+E_{i}v-X_{i}Sv\nonumber \\
	= & A\overline{x}_{i}+B\overline{u}_{i},\label{eq:16}\\
	e_{i}= & C_{i}(\overline{x}_{i}+X_{i}v)+D_{i}(\overline{u}_{i}+U_{i}v)+F_{i}v\nonumber \\
	= & C_{i}\overline{x}_{i}+D_{i}\overline{u}_{i}.\label{eq:58}
	\end{align}
	By (\ref{eq:10}), $\overline{u}_{i}$ is expressed as: 
	\begin{align*}
	\overline{u}_{i}= & K_{i}\hat{x}_{i}+(\hat{U}_{i}-K_{i}\hat{X}_{i})\hat{v}_{i}-U_{i}v\\
	= & K_{i}\overline{x}_{i}+K_{i}\tilde{x}_{i}+(\tilde{U}_{i}-K_{i}\tilde{X}_{i})\hat{v}_{i}+(U_{i}-K_{i}X_{i})\tilde{v}_{i}.
	\end{align*}
	Then (\ref{eq:16}) is written as:
	\begin{align*}
	\dot{\overline{x}}_{i}= & (A_{i}+B_{i}K_{i})\overline{x}_{i}\\
	& +K_{i}\tilde{x}_{i}+(\tilde{U}_{i}-K_{i}\tilde{X}_{i})\hat{v}_{i}+(U_{i}-K_{i}X_{i})\tilde{v}_{i}.
	\end{align*}
	Note that $\tilde{x}_{i},\tilde{v}_{i},\tilde{X}_{i},\tilde{U}_{i}$
	all converge to zero exponentially, and $A_{i}+B_{i}K_{i}$ is a Hurwitz
	matrix. Thus, $\overline{x}_{i}$ will converge to zero exponentially.
	Then $e_{i}$ will converge to zero exponentially. The proof is completed.
\end{proof}

\section{\textcolor{black}{Proof of Theorem \ref{thm:2}}}
\begin{proof}
In the following let $\mathrm{c}_{j}(j=1,2,...,18)$ denote some proper
positive constants. $\mathrm{\varphi}_{ij}(\overline{\iota}_{\psi},\overline{\iota}_{\omega},\mathcal{T}^{i})(j=1,2,...,8)$
denote some non-negative increasing functions for agent $i$ with
$\mathrm{\varphi}_{ij}(0,0,0)=0$. $\zeta_{1},\zeta_{2},\zeta_{3},\varepsilon$
are some positive design parameters. The proof is then divided into
the following steps.
	
	\textit{Step 1).} We will find the relation between $\tilde{x}_{i}(t)$
	and $\tilde{x}_{i}(\tau_{p}^{i})$.
	
	Using (\ref{eq:17-2-1}) and Theorem \ref{thm:1}, we can obtain that
	\begin{align}
	||E_{i}\tilde{v}_{i}+L_{i}F_{mi}\left(\overline{v}_{i}(\tau_{p}^{i},\overline{t}_{l}^{i})-v(\tau_{p}^{i})\right)\nonumber \\
	+L_{i}\left(\psi_{i}(\overline{\tau}_{q}^{i})-\psi_{i}(\tau_{p}^{i})\right)|| & <\mathrm{c}_{1}\mathrm{e}^{-\mathrm{c}_{2}t}+||L_{i}||\overline{\iota}_{\psi}.\label{eq:15-1-1}
	\end{align}
	
	Then by following the line of Step 1) in the proof of Theorem \ref{thm:1-1}
	and using the above inequality, we get
	\begin{align}
	||\tilde{x}_{i}(t)-\tilde{x}_{i}(\tau_{p}^{i})||\leq & \delta_{i3}(\mathcal{T}^{i})||\tilde{x}_{i}(\tau_{p}^{i})||+\mathcal{T}^{i}\mathrm{e}^{||A_{i}||\mathcal{T}^{i}}\mathrm{c}_{1}\mathrm{e}^{-\mathrm{c}_{2}\tau_{p}^{i}}\nonumber \\
	& +\mathcal{T}^{i}||L_{i}||\overline{\iota}_{\psi}\mathrm{e}^{||A_{i}||\mathcal{T}^{i}}.
	\end{align}
	Then when $\mathcal{T}^{i}$ is small enough, we have
	\begin{align}
	||\tilde{x}_{i}-\tilde{x}_{i}(\tau_{p}^{i})||\leq & \delta_{\tilde{x}}(\mathcal{T}^{i})||\tilde{x}_{i}||+\mathrm{c}_{3}\mathrm{e}^{-\mathrm{c}_{2}t}+\mathrm{\varphi}_{i1}(\cdot)\label{eq:60}
	\end{align}
	where
	\begin{equation}
	\mathrm{\varphi}_{i1}(\overline{\iota}_{\psi},\overline{\iota}_{\omega},\mathcal{T}^{i})=\frac{\mathcal{T}^{i}||L_{i}||\overline{\iota}_{\psi}\mathrm{e}^{||A_{i}||\mathcal{T}^{i}}}{1-\delta_{i3}(\mathcal{T}^{i})}.\label{eq:phi1}
	\end{equation}
	
	\textit{Step 2).} We will show $\tilde{x}_{i}$ will converge to a
	small neighborhood of origin.
	
	Consider the Lyapunov function $V_{\tilde{x}}=\frac{1}{2}\tilde{x}_{i}^{T}Q_{i}\tilde{x}_{i}.$
	Based on (\ref{eq:55}), (\ref{eq:60}), (\ref{eq:17-2-1}) and Theorem
	\ref{thm:1}, we obtain
	\begin{align*}
	\dot{V}_{\tilde{x}}\leq & -||\tilde{x}_{i}||^{2}+||Q_{i}L_{i}C_{mi}||\delta_{\tilde{x}}(\mathcal{T}^{i})||\tilde{x}_{i}||^{2}\\
	& +||\tilde{x}_{i}||\cdot||Q_{i}L_{i}C_{mi}||\varphi_{i1}(\mathcal{T}^{i},\overline{\iota}_{\psi})+||\tilde{x}_{i}||\cdot||Q_{i}L_{i}||\overline{\iota}_{\psi}\\
	& +||\tilde{x}_{i}||\mathrm{c}_{4}\mathrm{e}^{-\mathrm{c}_{5}t}.
	\end{align*}
	By Young's inequality, we have
	\begin{align}
	\dot{V}_{\tilde{x}}\leq & -(1-||Q_{i}L_{i}C_{mi}||\delta_{\tilde{x}}(\mathcal{T}^{i})-\zeta_{1}-\varepsilon)||\tilde{x}_{i}||^{2}\nonumber \\
	& +\mathrm{c}_{6}\mathrm{e}^{-2\mathrm{c}_{5}t}+\varphi_{i2}\label{eq:62-2}
	\end{align}
	where $\varepsilon>0$ is a small design parameter such that $||Q_{i}L_{i}C_{mi}||\delta_{\tilde{x}}(\mathcal{T}^{i})+\zeta_{1}+\varepsilon<1$,
	\begin{equation}
	\mathrm{\varphi}_{i2}(\overline{\iota}_{\psi},\overline{\iota}_{\omega},\mathcal{T}^{i})=\frac{\left(||Q_{i}L_{i}C_{mi}||\varphi_{i1}(\mathcal{T}^{i},\overline{\iota}_{\psi})+||Q_{i}L_{i}||\overline{\iota}_{\psi}\right)^{2}}{4\zeta_{1}}.\label{eq:phi2}
	\end{equation}
	Note that 
	\begin{equation}
	\frac{\lambda_{\min}(Q)||\tilde{x}_{i}||^{2}}{2}\leq V_{\tilde{x}}\leq\frac{\lambda_{\max}(Q)||\tilde{x}_{i}||^{2}}{2}.\label{eq:63-2}
	\end{equation}
	Therefore, (\ref{eq:62-2}) can be expressed as:
	\begin{align}
	\dot{V}_{\tilde{x}}\leq & -(1-||Q_{i}L_{i}C_{mi}||\delta_{\tilde{x}}(\mathcal{T}^{i})-\zeta_{1}-\varepsilon)\frac{2V_{\tilde{x}}}{\lambda_{\max}(Q)}\nonumber \\
	& +\mathrm{c}_{6}\mathrm{e}^{-2\mathrm{c}_{5}t}+\varphi_{i2}.\label{eq:62-2-1}
	\end{align}
	
	By solving the above inequality, we can conclude that 
	\begin{align}
	V_{\tilde{x}}\leq\mathrm{c}_{7}\mathrm{e}^{-\mathrm{c}_{8}t}+ & \frac{\lambda_{\max}(Q)\varphi_{i2}}{2(1-||Q_{i}L_{i}C_{mi}||\delta_{\tilde{x}}(\mathcal{T}^{i})-\zeta_{1}-\varepsilon)}.\label{eq:61-2}
	\end{align}
	Using (\ref{eq:63-2}), we have
	\begin{equation}
	||\tilde{x}_{i}||\leq\frac{\mathrm{2c}_{7}\mathrm{e}^{-\mathrm{c}_{8}t}}{\lambda_{\min}(Q)}+\varphi_{i3}\label{eq:61}
	\end{equation}
	where
	\begin{equation}
	\varphi_{i3}=\frac{\lambda_{\max}(Q)\varphi_{i2}}{\lambda_{\min}(Q)\left(1-||Q_{i}L_{i}C_{mi}||\delta_{\tilde{x}}(\mathcal{T}^{i})-\zeta_{1}-\varepsilon\right)}.\label{eq:phi3}
	\end{equation}
	This shows that $\tilde{x}_{i}$ will converge to a small neighborhood
	of origin.
	
	\textit{Step 3).} We will find the relation between $\overline{x}_{i}$
	and $\overline{x}_{i}(\tau_{p}^{i})$.
	
	Using (\ref{eq:10-2}) and (\ref{eq:12-2}), $\overline{u}_{i}$ is
	expressed as:
	\begin{align}
	\overline{u}_{i}= & \omega_{i}(\tau_{p}^{i})+\omega_{i}(\overline{\varsigma}_{m}^{i})-\omega_{i}(\tau_{p}^{i})\nonumber \\
	= & K_{i}\hat{x}_{i}(\tau_{p}^{i})+(\hat{U}_{i}(\tau_{p}^{i})-K_{i}\hat{X}_{i}(\tau_{p}^{i}))\hat{v}_{i}(\tau_{p}^{i})-U_{i}v\nonumber \\
	& +\omega_{i}(\overline{\varsigma}_{m}^{i})-\omega_{i}(\tau_{p}^{i})\nonumber \\
	= & K_{i}\overline{x}_{i}(\tau_{p}^{i})+\mathcal{U}_{i}\label{eq:63-1}
	\end{align}
	where
	\begin{align*}
	\mathcal{U}_{i}= & K_{i}\tilde{x}_{i}(\tau_{p}^{i})+(\tilde{U}_{i}(\tau_{p}^{i})-K_{i}\tilde{X}_{i}(\tau_{p}^{i}))\hat{v}_{i}(\tau_{p}^{i})\\
	+ & (U_{i}-K_{i}X_{i})\tilde{v}_{i}(\tau_{p}^{i})+U_{i}(v(\tau_{p}^{i})-v(t))\\
	+ & (\omega_{i}(\overline{\varsigma}_{m}^{i})-\omega_{i}(\tau_{p}^{i})).
	\end{align*}
	
	Substituting the above equation into (\ref{eq:16}), we obtain
	\begin{align}
	\dot{\overline{x}}_{i}= & A_{i}\overline{x}_{i}+B_{i}K_{i}\overline{x}_{i}(\tau_{p}^{i})+B_{i}\mathcal{U}_{i}\label{eq:62}
	\end{align}
	Note that there exists a non-negative constant $\zeta_{2}$ such that
	\begin{equation}
	||v(\tau_{p}^{i})-v(t)||\leq\zeta_{2}\mathcal{T}^{i}\label{eq:aa}
	\end{equation}
	where $\zeta_{2}=0$ when $v(t)$ is a constant signal.
	
	Meanwhile, $\tilde{v}_{i},\tilde{X}_{i},\tilde{U}_{i}$ all converge
	to zero exponentially. Then by (\ref{eq:61}) and (\ref{eq:26-1-1-1}),
	we conclude that 
	\begin{equation}
	||\mathcal{U}_{i}||\leq\mathrm{c}_{9}\mathrm{e}^{-\mathrm{c}_{10}t}+\varphi_{i4}(\overline{\iota}_{\psi},\overline{\iota}_{\omega},\mathcal{T}^{i})\label{eq:63}
	\end{equation}
	where 
	\begin{equation}
	\varphi_{i4}(\overline{\iota}_{\psi},\overline{\iota}_{\omega},\mathcal{T}^{i})=||K_{i}\varphi_{i3}||+\overline{\iota}_{\omega}+||U_{i}\zeta_{2}\mathcal{T}^{i}||.\label{eq:phi4}
	\end{equation}
	
	Using this for (\ref{eq:62}) and by Gronwall's inequality, we obtain
	\begin{align}
	||\overline{x}_{i}-\overline{x}_{i}(\tau_{p}^{i})||\leq & \delta_{i4}(\mathcal{T}^{i})||\overline{x}_{i}(\tau_{p}^{i})||+\mathcal{T}^{i}\mathrm{e}^{||A_{i}||\mathcal{T}^{i}}\mathrm{c}_{9}\mathrm{e}^{-\mathrm{c}_{10}\tau_{p}^{i}}\nonumber \\
	& +\mathcal{T}^{i}\mathrm{e}^{||A_{i}||\mathcal{T}^{i}}\varphi_{i4}(\overline{\iota}_{\psi},\overline{\iota}_{\omega},\mathcal{T}^{i})
	\end{align}
	where 
	\begin{equation}
	\delta_{i4}(\mathcal{T}^{i})=\mathcal{T}^{i}||A_{i}+B_{i}K_{i}||\mathrm{e}^{||A_{i}||\mathcal{T}^{i}}.\label{eq:64-1}
	\end{equation}
	Then when $\mathcal{T}^{i}$ is small enough, we have
	\begin{align}
	||\overline{x}_{i}-\overline{x}_{i}(\tau_{p}^{i})|| & \leq\delta_{\overline{x}}(\mathcal{T}^{i})||\overline{x}_{i}||+\mathrm{c}_{11}\mathrm{e}^{-\mathrm{c}_{12}t}+\varphi_{i5}\label{eq:64}
	\end{align}
	where 
	\begin{equation}
	\delta_{\overline{x}}(\mathcal{T}^{i})=\frac{\delta_{i4}(\mathcal{T}^{i})}{1-\delta_{i4}(\mathcal{T}^{i})},\label{eq:34-1-1-1}
	\end{equation}
	\begin{equation}
	\varphi_{i5}=\frac{\mathcal{T}^{i}\mathrm{e}^{||A_{i}||\mathcal{T}^{i}}\varphi_{i4}(\overline{\iota}_{\psi},\overline{\iota}_{\omega},\mathcal{T}^{i})}{1-\delta_{i4}(\mathcal{T}^{i})}.\label{eq:phi5}
	\end{equation}
\textit{Step 4).} We will show $\overline{x}_{i},e_{i}$ converge
to a small neighborhood of origin.

Consider 
\[
V_{\overline{x}}=\frac{1}{2}\overline{x}_{i}^{T}R_{i}\overline{x}_{i}
\]
where $R_{i}$ is a positive matrix such that $R_{i}(A_{i}+B_{i}K_{i})+(A_{i}+B_{i}K_{i})^{\mathrm{\mathit{T}}}R_{i}=2I$. 

Then using (\ref{eq:62}),
\begin{align*}
\dot{V}_{\overline{x}}= & -||\overline{x}_{i}||^{2}+\overline{x}_{i}^{T}R_{i}B_{i}K_{i}(\overline{x}_{i}(\tau_{p}^{i})-\overline{x}_{i})+\overline{x}_{i}^{T}R_{i}\mathcal{U}_{i}.
\end{align*}
By (\ref{eq:64}), (\ref{eq:63}) and Young's inequality, we get
\begin{align}
\dot{V}_{\overline{x}}\leq & -(1-||R_{i}B_{i}K_{i}||\delta_{\overline{x}}(\mathcal{T}^{i})-\zeta_{3}-\varepsilon)||\overline{x}_{i}||^{2}\nonumber \\
& +\mathrm{c}_{13}\mathrm{e}^{-\mathrm{c}_{14}t}+\varphi_{i6}(\overline{\iota}_{\psi},\overline{\iota}_{\omega},\mathcal{T}^{i})\label{eq:74}
\end{align}
\textcolor{blue}{where $\varepsilon>0$ is a small design parameter
	such that $||R_{i}B_{i}K_{i}||\delta_{\overline{x}}(\mathcal{T}^{i})+\zeta_{3}+\varepsilon<1$,}
\begin{equation}
\varphi_{i6}(\overline{\iota}_{\psi},\overline{\iota}_{\omega},\mathcal{T}^{i})=\frac{\left(||R_{i}B_{i}K_{i}\varphi_{i5}||+||R_{i}\varphi_{i4}||\right)^{2}}{4\zeta_{3}}.\label{eq:phi6}
\end{equation}

Note that 
\begin{equation}
\frac{\lambda_{\min}(R_{i})||\overline{x}_{i}||^{2}}{2}\leq V_{\overline{x}}\leq\frac{\lambda_{\max}(R_{i})||\overline{x}_{i}||^{2}}{2}.\label{eq:63-2-1}
\end{equation}
Therefore, (\ref{eq:74}) can be expressed as:
\begin{align}
\dot{V}_{\overline{x}}\leq & -(1-||R_{i}B_{i}K_{i}||\delta_{\overline{x}}(\mathcal{T}^{i})-\zeta_{3}-\varepsilon)\frac{2V_{\overline{x}}}{\lambda_{\max}(R_{i})}\nonumber \\
& +\mathrm{c}_{13}\mathrm{e}^{-\mathrm{c}_{14}t}+\varphi_{i6}(\overline{\iota}_{\psi},\overline{\iota}_{\omega},\mathcal{T}^{i}).\label{eq:74-1}
\end{align}

By solving the above equation, we obtain 
\begin{align}
||\overline{x}_{i}||\leq & \mathrm{c}_{15}\mathrm{e}^{-\mathrm{c}_{16}t}+\varphi_{i7}(\overline{\iota}_{\psi},\overline{\iota}_{\omega},\mathcal{T}^{i})\label{eq:61-1-2}
\end{align}
where\textcolor{black}{
	\begin{equation}
	\varphi_{i7}=\sqrt{\frac{\lambda_{\max}(R_{i})\varphi_{i6}}{\lambda_{\min}(R_{i})\left(1-||R_{i}B_{i}K_{i}||\delta_{\overline{x}}(\mathcal{T}^{i})-\zeta_{3}-\varepsilon\right)}}.\label{eq:phi7}
	\end{equation}
}This means that $||\overline{x}_{i}||$ will converge to a small
neighborhood around origin.

Finally, noting (\ref{eq:58}), (\ref{eq:63-1}) and (\ref{eq:63}),
we can conclude that
\begin{align}
||e_{i}||\leq & \mathrm{c}_{17}\mathrm{e}^{-\mathrm{c}_{18}t}+\varphi_{i8}(\overline{\iota}_{\psi},\overline{\iota}_{\omega},\mathcal{T}^{i})\label{eq:61-1-1}
\end{align}
where\textcolor{blue}{
	\begin{equation}
	\varphi_{i8}(\overline{\iota}_{\psi},\overline{\iota}_{\omega},\mathcal{T}^{i})=||C_{i}+D_{i}K_{i}||\varphi_{i7}+||D_{i}K_{i}||\varphi_{i3}.\label{eq:phi8}
	\end{equation}
}Therefore, we know $\underset{t\rightarrow+\infty}{\lim}||e_{i}(t)|| \leq\varLambda_{i}\leq\varphi_{i8}(\overline{\iota}_{\psi},\overline{\iota}_{\omega},\mathcal{T}^{i})$
where $\varphi_{i8}(\overline{\iota}_{\psi},\overline{\iota}_{\omega},\mathcal{T}^{i})$
is an increasing function with $\varphi_{i8}(0,0,0)=0$. $\varphi_{i8}(\overline{\iota}_{\psi},\overline{\iota}_{\omega},\mathcal{T}^{i})$
can be computed by using the values of $\varphi_{ij}(\overline{\iota}_{\psi},\overline{\iota}_{\omega},\mathcal{T}^{i})(j=1,2,...,7)$
in (\ref{eq:phi1}), (\ref{eq:phi2}), (\ref{eq:phi3}), (\ref{eq:phi4}),
(\ref{eq:phi5}), (\ref{eq:phi6}), (\ref{eq:phi7}). This completes
the proof. 
\end{proof}
\begin{rem}
\textcolor{blue}{Note that in order to compute $\varphi_{i7}$ and
	$\varphi_{i8}$ in (\ref{eq:phi7}) and (\ref{eq:phi8}), one needs
	to select the design parameter $\varepsilon$, which is a very
	small constant. A smaller $\varepsilon$ will result in  smaller
	$\varphi_{i7}$ and $\varphi_{i8}$. This will have a tighter bound
	for  $\varLambda_{i}$.}
\end{rem}
\begin{rem}
From (\ref{eq:aa}), we know $\zeta_{2}=0$ when $v(t)$ is a constant
signal. Using this, we can conclude that $\varphi_{i8}(\overline{\iota}_{\psi},\overline{\iota}_{\omega},\mathcal{T}^{i})=0$
if $\overline{\iota}_{\psi}=0,\overline{\iota}_{\omega}=0$ and the
signal $v$ is a constant. This implies that we can make the regulation
error converge to exact zero for constant $v$ even if PETM-C is invoked.
\end{rem}

\section{Additional discussions}

There are several ways to remove the communication from controller
to sensor. One simple way is to modify the event-triggered condition
(\ref{eq:26-1-1}) into 
\begin{equation}
\overline{\tau}_{p+1}^{i}=\mathrm{inf}\{\tau>\overline{\tau}_{q}^{i}|\tau\in\Omega_{T},f_{\psi}^{i}(\cdot)>0\}\label{eq:26-1-1-2}
\end{equation}
where 
\begin{equation}
f_{\psi}^{i}(\tau,\overline{\tau}_{q}^{i})=||y_{mi}(\tau)-y_{mi}(\overline{\tau}_{q}^{i})||-\iota_{\psi}\mathrm{e}^{-\gamma_{\psi}\tau}-\overline{\iota}_{\psi}\label{eq:17-2-2}
\end{equation}
with positive constants $\iota_{\psi},\gamma_{\psi},\overline{\iota}_{\psi}>0$.
Then, the communication burden can be reduced be increasing the constant
$\overline{\iota}_{\psi}$ with a sacrifice of the control accuracy.
That is the regulation error $e_{i}(t)$ converges to an arbitrary
small neighborhood around origin.

Another way it to utilize the event-triggered control scheme shown
in Fig. \ref{fig:3-1-1-1} instead of Fig. \ref{fig:3-1}. We can
see that the sensor in each agent sends/receives the information to/from
its neighbors. Then the information $\hat{v}_{i}(\overline{t}_{l}^{i}),\hat{S}_{i}(\overline{t}_{l}^{i})$
can be directly used for PETM-B. Thus, the controller does not need
to send information to the sensor. However, the computational burden
may increase for the sensor side since the distributed observer should
be implemented in the sensor side to generate $\hat{v}_{i}(\overline{t}_{l}^{i}),\hat{S}_{i}(\overline{t}_{l}^{i})$. 

\begin{figure}
	\begin{centering}
		\includegraphics[scale=0.7]{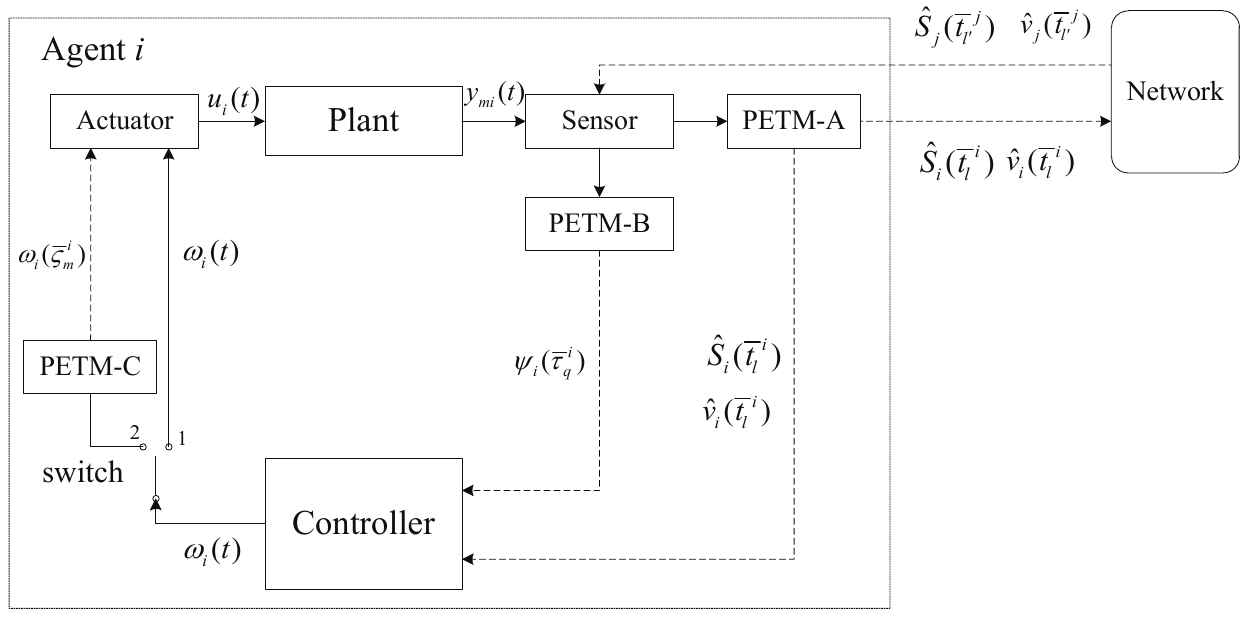}
		\par\end{centering}
	\caption{\label{fig:3-1-1-1}Modified control scheme.}
\end{figure}

\ \\
%


\section*{References}
\begin{thebibliography}{10}
	\bibitem[Behera, Bandyopadhyay, \& Yu(2018)]{key-8}   Behera, A., Bandyopadhyay, B., \& Yu, X. (2018).
	Periodic event-triggered sliding mode control. \textit{Automatica},
	\textit{96}, 61-72.
	
	\bibitem[Bernuau, Moulay, Coirault, \& Isfoula(2018)]{key-3a-1} Bernuau, E., Moulay, E., Coirault, P., \& Isfoula,
	F. (2019). Practical consensus of homogeneous sampled-data multi-agent
	systems. \textit{IEEE Transactions on Automatic Control}, DOI: 10.1109/TAC.2019.2904442.
	
	\bibitem[Bin, Marconi, \& Teel(2019)]{key-13} Bin, M., Marconi, L., \& Teel, A. (2019). Adaptive
	output regulation for linear systems via discrete-time identifiers.
	\textit{Automatica}, \textit{105}, 422-432.
	
	\bibitem[Cai, \& Hu(2019)]{key-16} Cai, H., \& Hu, G. (2019). Dynamic consensus tracking
	of uncertain Lagrangian systems with a switched command generator.
	\textit{IEEE Transactions on Automatic Control}, DOI: 10.1109/TAC.2019.2893874.
	
	\bibitem[Cai, \& Huang(2016)]{key-12} Cai, H., \& Huang, J. (2016). The leader-following
	consensus for multiple uncertain Euler-Lagrange systems with an adaptive
	distributed observer. \textit{IEEE Transactions on Automatic Control},
	\textit{61}(10), 3152-3157. 
	
	\bibitem[Cai, Lewis, Hu, \& Huang(2017)]{key-14aa} Cai, H., Lewis, F., Hu, G., \& Huang, J. (2017).
	The adaptive distributed observer approach to the cooperative output
	regulation of linear multi-agent systems. \textit{Automatica}, \textit{61},
	299-305.
	
	\bibitem[Chen, \& Chen(2017)]{key-ad1} Chen, X., \& Chen, Z. (2017). Robust sampled-data
	output synchronization of nonlinear heterogeneous multi-agents. \textit{IEEE
		Transactions on Automatic Control}, \textit{62}(3), 1458-1464. 
	
	\bibitem[Chen, \& Huang(2015)]{key-14a} Chen, Z., \& Huang, J. (2015). \textit{Stabilization
		and regulation of nonlinear systems: A robust and adaptive approach}.
	Springer International Publishing Switzerland. 
	
	\bibitem[Chen, \& Sun(2020)]{key-9-2}  Chen, C., \& Sun, Z. (2020). A unified approach
	to finite-time stabilization of high-order nonlinear systems with
	an asymmetric output constraint. \textit{Automatica}, \textit{111},
	108581. 
	
	\bibitem[Cheng, \& Ugrinovskii(2016)]{key-14aa-1} Cheng, Y., \& Ugrinovskii, V. (2016). Event-triggered
	leader-following tracking control for multivariable multi-agent systems.
	\textit{Automatica}, \textit{70}, 204-210.
	
	\bibitem[Cuenca, Antunes, Castillo, Garcia, Khashooei, \& Heemels(2019)]{key-10-1}  Cuenca, A., Antunes, D., Castillo, A., Garcia,
	P., Khashooei, B., \& Heemels, W. (2019). Periodic event-triggered
	sampling and dual-rate control for a wireless networked control system
	with applications to UAVs. \textit{IEEE Transactions on Industrial
		Electronics}, \textit{66}(4), 3157-3166.
	
	\bibitem[Deng, \& Yang(2019)]{key-9} Deng, C., \& Yang, G. (2019). Distributed adaptive
	fault-tolerant control approach to cooperative output regulation for
	linear multi-agent systems. \textit{Automatica}, \textit{103}, 62-68.
	
	\bibitem[Franceschelli, Gasparri, Giua, \& Seatzu(2013)]{key-8a} Franceschelli, M., Gasparri, A., Giua, A., \& Seatzu,
	C. (2013). Decentralized estimation of Laplacian eigenvalues in multi-agent
	systems. \textit{Automatica}, 49(4), 1031-1036.
	
	\bibitem[Garcia, Cao, \& Casbeer(2017)]{key-9-1} Garcia, E., Cao, Y., \& Casbeer, D. (2017). Periodic
	event-triggered synchronization of linear multi-agent systems with
	communication delays. \textit{IEEE Transactions on Automatic Control},
	62(1), 366-371.
	
	\bibitem[Hu, Liu, \& Feng(2018)]{key-5} Hu, W., Liu, L., \& Feng, G. (2018). Cooperative
	output regulation of linear multi-agent systems by intermittent communication:
	A unified framework of time and event-triggering strategies. \textit{IEEE
		Transactions on Automatic Control}, \textit{63}(2), 548-555. 
	
	\bibitem[Hu, Liu, \& Feng(2019)]{key-6} Hu, W., Liu, L., \& Feng, G. (2019). Event-triggered
	cooperative output regulation of linear multi-agent systems under
	jointly connected topologies. \textit{IEEE Transactions on Automatic
		Control}, \textit{64}(3), 1317-1322. 
	
	\bibitem[Li, Li, \& Tong(2019)]{key-li-1} Li, Y., Li, K., \& Tong, S. (2019). Finite-time
	adaptive fuzzy output feedback dynamic surface control for MIMO non-strict
	feedback systems. \textit{IEEE Transactions on Fuzzy Systems}, vol.27,
	no.1, pp.96-110, 2019. 
	
	\bibitem[Li, Qu, \& Tong(2019)]{key-li-2}  Li, Y., Qu, F., \& Tong, S. (2019). Observer-based
	fuzzy adaptive finite time containment control of nonlinear multi-agent
	systems with input-delay. \textit{IEEE Transactions on Cybernetics,
	}DOI: 10.1109/TCYB.2020.2970454. 

\bibitem[Li, Xing, Zhao, \& Shi(2017)]{key-icl} Li, L., Xing, W., Zhao, Y.,  \& Shi, P. (2017). Stability analysis of multi-agent systems with multiple leaders of variable velocities based on consensus protocols. \textit{ ICIC Express Letters}, \textit{11}(11), 1599-1609.

	
	\bibitem[Liu, \& Huang(2018)]{key-6a} Liu, W., \& Huang, J. (2018). Cooperative global
	robust output regulation for a class of nonlinear multi-agent systems
	by distributed event-triggered control. \textit{Automatica}, \textit{93},
	138-148.
	
	\bibitem[Liu, \& Huang(2019)]{key-6aa} Liu, W., \& Huang, J. (2019). Leader-following
	consensus for linear multi-agent systems via asynchronous sampled-data
	control. \textit{IEEE Transactions on Automatic Control}, DOI: 10.1109/TAC.2019.2948256. 
	
	\bibitem[Liu, \& Huang(2020)]{key-6aaa} Liu, W., \& Huang, J. (2020). Output regulation
	of linear systems via sampled-data control. \textit{Automatica}, \textit{113},
	108684.
	
	\bibitem[Liu, \& Huang(2019)]{key-10} Liu, T., \& Huang, J. (2019). A distributed observer
	for a class of nonlinear systems and its application to a leader-following
	consensus problem. \textit{IEEE Transactions on Automatic Control},
	\textit{64}(3), 1221-1227. 
	
	\bibitem[Liu, Ma, Lewis, \& Wan(2019)]{key-10aa}  Liu, H., Ma, T., Lewis, F., \& Wan, Y. (2019).
	Robust formation trajectory tracking control for multiple quadrotors
	with communication delays. \textit{IEEE Transactions on Control Systems
		Technology}, DOI: 10.1109/TCST.2019.2942277.
	
	\bibitem[Menard, Moulay, Coirault, \& Defoort(2019)]{key-20} Menard, T., Moulay, E., Coirault, P., \& Defoort,
	M. (2019). Observer-based consensus for second-order multi-agent systems
	with arbitrary asynchronous and aperiodic sampling periods. \textit{Automatica},
	\textit{99}, 237-245.
	
	\bibitem[Meng, Xie, \& Soh(2017)]{key-11} Meng, X., Xie, L., \& Soh Y. (2017). Asynchronous
	periodic event-triggered consensus for multi-agent systems\textit{.
		Automatica}, \textit{84}, 214-220.
	
	\bibitem[Nowzari, Garcia, \& Cortés(2019)]{key-2a-1} Nowzari, C., Garcia E., \& Cortés, J. (2019).
	Event-triggered communication and control of networked systems for
	multi-agent consensus\textit{. Automatica}, \textit{105}, 1-27.
	
	\bibitem[Poveda, \& Teel(2019)]{key-b1} Poveda, J., \& Teel, A. (2019). Hybrid mechanisms
	for robust synchronization and coordination of multi-agent networked
	sampled-data systems. \textit{Automatica}, 99, 41-53.
	
	\bibitem[Shi,  \& Shen(2017)]{key-b2} Shi, P., \& Shen, Q. (2017). Observer-based leader-following
	consensus of uncertain nonlinear multi-agent systems. \textit{International
		Journal of Robust and Nonlinear Control}, 27(17), 3794-3811.
	

   \bibitem[Singh, Tiwari, Garg(2018)]{key-b33} Singh, A., Tiwari, V., \& Garg, P. (2018). Anonymous decision logic for secure multi-agent computing. 	\textit{ICIC Express Letters Part B Applications}, \textit{9}(12), 1217-1222.

	\bibitem[Stankovic, Stankovic, \& Johansson(2018)]{key-b4} Stankovic, M., Stankovic, S., \& Johansson, K. (2018).
	Distributed time synchronization for networks with random delays and
	measurement noise. \textit{Automatica}, 93, 126-137.
	
	\bibitem[Su(2019)]{key-2} Su, Y. (2019). Semi-global output feedback cooperative
	control for nonlinear multi-agent systems via internal model approach.
	\textit{Automatica}, \textit{103}, 200-207.
	
	\bibitem[Sun, Hu, Xie, \& Egerstedt(2018)]{key-2a} Sun, C., Hu, G., Xie, L., \& Egerstedt, M. (2018).
	Robust finite-time connectivity preserving coordination of second-order
	multi-agent systems. \textit{Automatica}, \textit{89}, 21-27.
	
	\bibitem[Wang, Postoyan, Nesic, \& Heemels(2019)]{key-7} Wang, W., Postoyan, R., Nesic, D., \& Heemels, W.
	(2019). Periodic event-triggered control for nonlinear networked control
	systems. \textit{IEEE Transactions on Automatic Control}, DOI: 10.1109/TAC.2019.2914255.
	
	\bibitem[Xing, Wen, Liu, Su, \& Cai(2017)]{key-7a} Xing, L., Wen, C, Liu, Z., Su, H., \& Cai, J. (2017).
	Event-triggered adaptive control for a class of uncertain nonlinear
	systems. \textit{IEEE Transactions on Automatic Control}, \textit{62}(4),
	2107-2076. 
	
		\bibitem[Xu, Yang, Wang, \& Shu(2019)]{key-inn} Xu, M., Yang, P., Wang, Y. \& Shu, Q. (2019). Observer-based multi-agent system fault upper bound estimation and fault-tolerant consensus control. \textit{International Journal of Innovative Computing, Information and Control}, \textit{15}(2), 519-534.
	
	\bibitem[Yang, \& Liberzon(2018)]{key-14} Yang, G., \& Liberzon, D. (2018). Feedback stabilization
	of switched linear systems with unknown disturbances under data-rate
	constraints. \textit{IEEE Transactions on Automatic Control}, \textit{63}(7),
	2107-2122. 
	
	\bibitem[Yang, Sun, Zheng, \& Li(2018)]{key-4} Yang, J., Sun, J., Zheng, W., \& Li, S. (2018). Periodic
	event-triggered robust output feedback control for nonlinear uncertain
	systems with time-varying disturbance. \textit{Automatica}, \textit{94},
	324-333.
	
	\bibitem[Yang, Zhang, Feng, Yan, \& Wang(2019)]{key-3} Yang, R., Zhang, H., Feng, G., Yan, H., \& Wang,
	Z. (2019). Robust cooperative output regulation of multi-agent systems
	via adaptive event-triggered control. \textit{Automatica}, \textit{102},
	129-136.
	
	\bibitem[Zhang, Gao, \& Kaynak(2013)]{key-14-1}  Zhang, L., Gao, H., \& Kaynak, O. (2013). Network-induced
	constraints in networked control systems-A survey. \textit{IEEE Transactions
		on Industrial Informatics}, \textit{9}(1), 403-416.
	
	\bibitem[Zheng, Shi, Wang, \& Shi(2019)]{key-15} Zheng, S., Shi, P., Wang, S., \& Shi, Y. (2019).
	Event-triggered adaptive fuzzy consensus for interconnected switched
	multiagent systems. \textit{IEEE Transactions on Fuzzy Systems}, \textit{27}(1),
	144-158.
	
		\bibitem[Zhu, \& Zheng(2019)]{key-16a} Zhu, Y., \& Zheng, W. (2020).
	Multiple Lyapunov functions analysis approach for discrete-time switched piecewise-affine systems under dwell-time constraints. \textit{IEEE Transactions on Automatic Control}, \textit{65}(5),
	2177-2184.
	
			\bibitem[Zhu, Zheng, \& Zhou(2019)]{key-17} Zhu, Y.,  Zheng, W., \& Zhou, D. (2020).
	Quasi-synchronization of discrete-time Lur'e-type switched systems with parameter mismatches and relaxed PDT constraints. \textit{IEEE Transactions on Cybernetics}, \textit{50}(5),
	 2026-2037.
\end{thebibliography}
\end{document}